\documentclass[submission,copyright,creativecommons]{eptcs}
\usepackage{graphicx}
\providecommand{\event}{GAM 2017} 
\usepackage{breakurl}             

\usepackage[latin1]{inputenc}
\usepackage{amsmath,amsthm}
\usepackage{enumitem}
\usepackage[all,2cell]{xy}
\CompileMatrices
\UseAllTwocells
\SilentMatrices
\usepackage{graphicx}
\usepackage{float}
\usepackage{fancybox}
\usepackage{enumitem}
\setenumerate{noitemsep,topsep=0pt,parsep=0pt,partopsep=0pt}

\usepackage{txfonts}
\usepackage{wrapfig}
\usepackage{mathtools}

\newtheorem*{claim}{Claim}
\theoremstyle{plain}
\newtheorem{theorem}{Theorem}[section]
\newtheorem{corollary}[theorem]{Corollary}
\newtheorem{proposition}[theorem]{Proposition}
\newtheorem{lemma}[theorem]{Lemma}

\theoremstyle{definition}
\newtheorem{definition}[theorem]{Definition}
\newtheorem{example}[theorem]{Example}
\newtheorem{remark}[theorem]{Remark}

\newcommand{\eql}[1]{\overset{\scriptstyle #1}{=}}

\def \catC {\mathbf{C}}
\def \catA {\mathbf{A}}

\newcommand{\tr}[1]{\xrightarrow{#1}}    
\newcommand{\tl}[1]{\xleftarrow{#1}}    

\newcommand{\ord}[1]{\overline{#1}}  

\newcommand{\copair}[1]{({#1})}    

\newcommand{\Cospan}[1]{\mathsf{Csp}(#1)}
\def \FINSET {\mathbf{FinSet}} 
\newcommand{\PROP}[1]{\mathbf{Prop}_{\scriptscriptstyle #1}} 
\def \SMC{\mathbf{SymCat}} 
\def \CAT{\mathbf{Cat}} 

\newcommand{\perm}[1]{\mathbf{Perm}_{\scriptscriptstyle #1}}
\newcommand{\freehyp}[1]{\mathbf{H}_{\scriptscriptstyle #1}}
\newcommand{\syntax}[1]{\mathbf{P}_{\scriptscriptstyle #1}}

\newcommand{\FTerm}[1]{\mathbf{FTerm}_{\scriptscriptstyle #1}}
\newcommand{\Hyp}[1]{\mathbf{Hyp}_{\scriptscriptstyle #1}}

\newcommand{\frob}[1]{\mathbf{Frob}_{\scriptscriptstyle #1}}

\newcommand{\synTosem}[1]{[\! [ #1 ]\! ]}
\newcommand{\frobTosem}[1]{[ #1 ]}
\newcommand{\allTosem}[1]{\langle\! \langle #1 \rangle \! \rangle}

\newcommand{\rewiring}[1]{\ulcorner #1 \urcorner}

\newcommand\symNet{\lower3pt\hbox{$\includegraphics[width=20pt]{graffles/symmetryalt.pdf}$}}
\newcommand\Idnet{\lower3pt\hbox{$\includegraphics[width=20pt]{graffles/id.pdf}$}}

\newcommand{\Sigmafrob}[1]{\Sigma^{\scriptscriptstyle #1}_{\scriptscriptstyle \frob{}}}

\newcommand\lccB{\lower5pt\hbox{$\includegraphics[width=25pt]{graffles/rccr.pdf}$}}
\newcommand\rccB{\lower5pt\hbox{$\includegraphics[width=25pt]{graffles/lccl.pdf}$}}
\newcommand\lccn{\lower5pt\hbox{$\includegraphics[width=20pt]{graffles/cup.pdf}$}}
\newcommand\rccn{\lower5pt\hbox{$\includegraphics[width=20pt]{graffles/cap.pdf}$}}
\def \df {\ \ensuremath{:\!\!=}\ }

\newcommand{\rring}[1]{\ensuremath{\mathbb{#1}}}

\newcommand{\N}{\rring{N}}

\newcommand{\pullbacktop}[4]{%

{#1} \ar@/_/[ddr]_{#4} \ar@/^/[drr]^{#2}%

\ar@{.>}[dr]|-{#3} \\}

\newcommand{\id}{id}

\newcommand{\cgr}[2][scale=0.45]{\raisebox{0.1em}{\begingroup
\setbox0=\hbox{\includegraphics[#1]{graffles/#2}}%
\parbox{\wd0}{\box0}\endgroup}}

\def \poi {\,\ensuremath{;}\,}
\def \df {\ensuremath{:=}}
\def \tns {\ensuremath{\oplus}}
\def \: {\colon}

\newcommand\Bmult{\lower4pt\hbox{$\includegraphics[width=17pt]{graffles/Bmult.pdf}$}}
\newcommand\Bcomult{\lower4pt\hbox{$\includegraphics[width=17pt]{graffles/Bcomult.pdf}$}}
\newcommand\Bunit{\lower4pt\hbox{$\includegraphics[width=14pt]{graffles/Bunit.pdf}$}}
\newcommand\Bcounit{\lower4pt\hbox{$\includegraphics[width=14pt]{graffles/Bcounit.pdf}$}}

\newcommand\Wmult{\lower4pt\hbox{$\includegraphics[width=17pt]{graffles/Wmult.pdf}$}}
\newcommand\Wcomult{\lower4pt\hbox{$\includegraphics[width=17pt]{graffles/Wcomult.pdf}$}}
\newcommand\Wunit{\lower4pt\hbox{$\includegraphics[width=14pt]{graffles/Wunit.pdf}$}}
\newcommand\Wcounit{\lower4pt\hbox{$\includegraphics[width=14pt]{graffles/Wcounit.pdf}$}}

\newcommand{\rrule}[2]{\ensuremath{( #1,#2 )}}

\newcommand{\node}{\lower0pt\hbox{$\includegraphics[width=6pt]{graffles/node.pdf}$}}
\newcommand{\hyperedge}{\lower2pt\hbox{$\includegraphics[width=25pt]{graffles/hyperedge.pdf}$}}
\newcommand{\ZeronetT}{\lower4pt\hbox{$\includegraphics[width=14pt]{graffles/idzerocircuit.pdf}$}}

\newcommand\idncircuit{\lower4pt\hbox{$\includegraphics[width=18pt]{graffles/idncircuit.pdf}$}}

\newcommand{\dlcorner}{{\ar@{}[dl]|(.8){\text{\large $\urcorner$}}}}
\newcommand{\drcorner}{{\ar@{}[dr]|(.8){\text{\large $\ulcorner$}}}}

\newcommand{\sg}{\!\lower1pt\hbox{$\includegraphics[width=8pt]{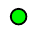}$}\!} 
\newcommand{\sr}{\!\lower1pt\hbox{$\includegraphics[width=8pt]{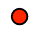}$}\!} 
\newcommand{\sbl}{\!\lower1pt\hbox{$\includegraphics[width=8pt]{graffles/blackbullet.pdf}$}\!} 

\makeatletter
\newcommand*{\@old@slash}{}\let\@old@slash\slash
\def\slash{\relax\ifmmode\delimiter"502F30E\mathopen{}\else\@old@slash\fi}
\makeatother

\def\backslash{\delimiter"526E30F\mathopen{}}

\newcommand{\Rew}[1]{\Rightarrow_{\! \scriptscriptstyle #1} }

\newcommand{\RS}{\mathcal{R}} 

\newcommand{\DPOstep}[1]{\rightsquigarrow_{\! \scriptscriptstyle #1} } 
\newcommand{\DPOstepL}[1]{\leftsquigarrow_{\! \scriptscriptstyle #1} }

\newcommand{\col}{\mathcal{C}}

\title{Rewriting in Symmetric Monoidal \\ and Hypergraph Categories}
\title{Rewriting in Free Hypergraph Categories}
\author{Fabio Zanasi
\institute{University College London, United Kingdom}
\email{f.zanasi@ucl.ac.uk}}
\def\titlerunning{Rewriting in Free Hypergraph Categories}
\def\authorrunning{F. Zanasi}

\begin{document}
\maketitle

\begin{abstract} We study rewriting for equational theories in the context of symmetric monoidal categories where there is a separable Frobenius monoid on each object. These categories, also called hypergraph categories, are increasingly relevant: Frobenius structures recently appeared in cross-disciplinary applications, including the study of quantum processes, dynamical systems and natural language processing. In this work we give a combinatorial characterisation of arrows of a free hypergraph category as cospans of labelled hypergraphs and establish a precise correspondence between rewriting modulo Frobenius structure on the one hand and double-pushout rewriting of hypergraphs on the other. This interpretation allows to use results on hypergraphs to ensure decidability of confluence for rewriting in a free hypergraph category. Our results generalise previous approaches where only categories generated by a single object (props) were considered.
\end{abstract}

\section{Introduction}\label{intro}

Symmetric monoidal categories (SMCs) are an increasingly popular mathematical framework for the formal analysis of network-style diagrammatic languages that are found across different disciplines. In an SMC arrows have two composition operations, intuitively corresponding to sequential ($c \poi d$) and parallel ($c \tns d$) combination of compound systems, and there are symmetry arrows $\cgr{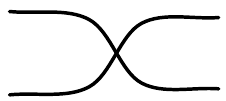}$, intuitively representing tangles of wires. These constructs are traditionally rendered by the two-dimensional notation of \emph{string diagrams}, which has the key advantage of absorbing most of the structural equalities prescribed by the definition of SMC. For instance, the two sides of the exchange law $(a_1 \poi a_2) \tns (b_1 \poi b_2) = (a_1 \tns b_1) \poi {(a_2 \tns b_2)}$ are encoded by the same string diagram $\cgr[width=1.5cm]{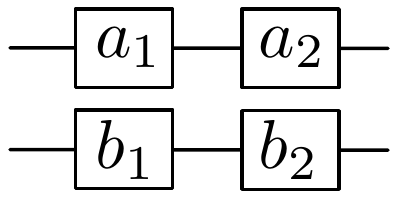}$. The graphical syntax emphasises \emph{connectivity} and \emph{resource-exchange} between components, which makes it particularly effective in the analysis of challenging computational models such as distributed systems (based on threads communication) and quantum processes (powered by a notion of non-separable---entangled---states).

Some applications demand SMCs with a richer structure. In this paper we focus on \emph{hypergraph categories}, which are SMCs where each object $x$ is equipped with a \emph{separable Frobenius monoid}. That means, for each $x$ there are operations as on the left, forming a commutative monoid and a commutative comonoid that interact according to the Frobenius  law and the separability law, as on the right.
\begin{equation}\label{eq:frobx}
\cgr{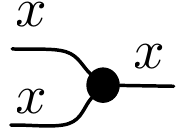} \qquad \cgr{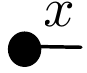} \qquad \cgr{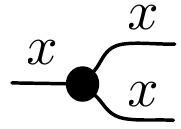} \qquad \cgr{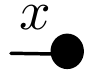} \qquad \qquad \qquad \cgr{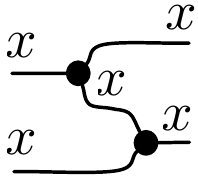} \ = \ \cgr{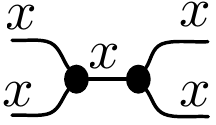} \ = \ \cgr{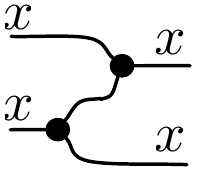} \qquad \cgr{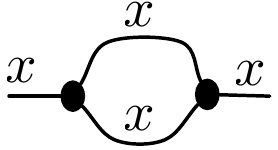} \ = \ \cgr{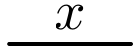}
\end{equation}
The core intuition is that this extra structure allows dangling wires of a string diagram to fork, be discarded, be moved to the left/right side, resulting in a more flexible manipulation of the interfaces (variables, memory cells) of the represented system. The use of hypergraph categories as an algebraic approach to computation was pioneered by Walters and collaborators \cite{Carboni1987,Katis1997a}, under the name of well-supported compact closed categories. Since then separable Frobenius monoids have appeared ubiquitously in diagrammatic calculi across diverse research threads. They notably feature in the ZX-calculus~\cite{Coecke2008} (quantum theory), where each Frobenius structure has a precise physical meaning in terms of quantum observables. Frobenius monoids also form the backbone of the calculus of stateless connectors \cite{Bruni2006}, the calculus of signal flow diagrams \cite{Bonchi2014b,BonchiSZ17}, Baez's network theory \cite{BaezErbele-CategoriesInControl} and Pavlovic's monoidal computer \cite{Pavlovic13}. More recently, a particular attention has been devoted to generic constructions of hypergraph categories through abstract notions of span, relation and their duals \cite{Zanasi16,Fon16,MG17,FZ-calco}.

Whereas separable Frobenius monoids constitute a common core for the aforementioned approaches, in each application string diagrams are further quotiented by domain-specific equations, instrumental in defining the appropriate notion of behavioural equality of systems. The perspective of this work is to acknowledge the conceptual difference between the symmetric monoidal and Frobenius structure on the one hand, which is a built-in part of any hypergraph category, and the domain-specific equations on the other hand. We shall study the latter as \emph{rewriting rules}: if the left hand side of such an equation can be found in a larger string diagram, it can be deleted and replaced with its right hand side.

This is coherent with the everyday practice of users of diagrammatic calculi and is the starting point for implementing graphical reasoning in a proof assistant. There is a thorough mathematical theory of rewriting for monoidal categories, which regards rewrite rules as generator 2-cells (variously called computads~\cite{Street-2cats} or polygraphs~\cite{Burroni1993}) and the possible rewriting trajectories as composite two-cells. However, this abstract perspective does not provide immediate help when it comes to \emph{implementing} rewriting. The main challenge is a concrete understanding of matching: in order to detect whether a string diagram contains the left-hand side of a rewriting rule, one needs to consider all its possible decompositions according to the structural equations. In an hypergraph category, this amounts to say that rewriting happens modulo the equations of separable Frobenius monoids. For instance, the rewriting rule 
$$\cgr[width=3.2cm]{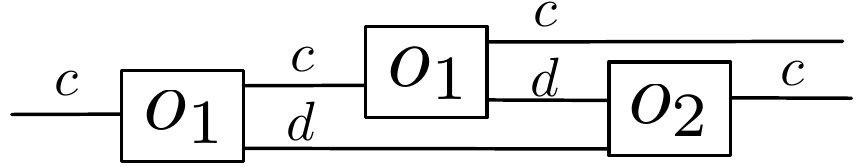}\quad  \Rew{} \quad \cgr[width=2.8cm]{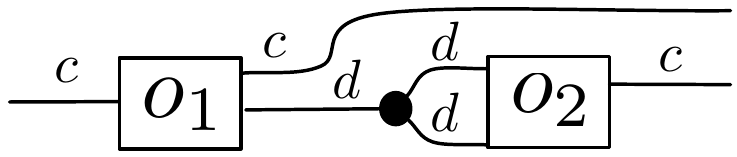}$$
applies in the leftmost string diagram below, module the separable Frobenius structure on objects $c, d$.
\begin{equation}\label{eq:matchingproblem}
\begin{aligned}
\cgr[width=4.2cm]{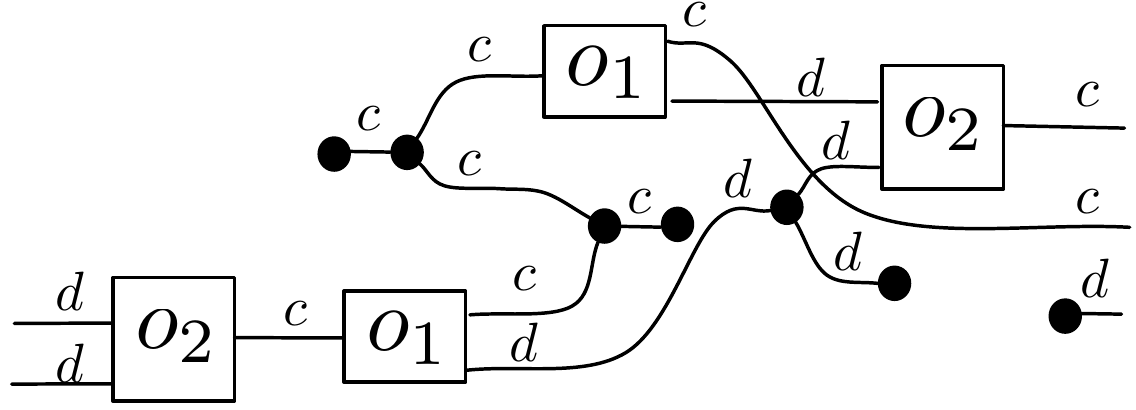} \ \eql{\eqref{eq:frobx}} \ \cgr[width=4.2cm]{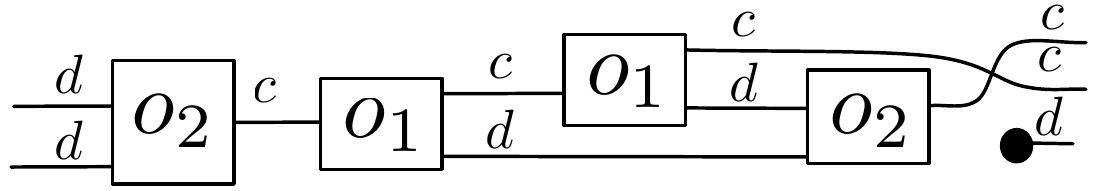} \quad \Rew{} \quad \cgr[width=4.2cm]{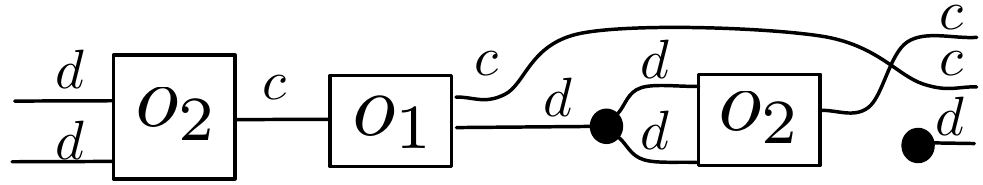}
\end{aligned}
\end{equation}
We shall develop a combinatorial characterisation of this rewriting mechanism, in three steps:
\begin{itemize}
\item the first step (Section \ref{sec:combinatorialmodel}) is to give a combinatorial description of the string diagrams in free hypergraph categories. Our choice is a category of cospans of hypergraphs, which is shown to be isomorphic to the free hypergraph category generated by a signature. This data structure encapsulates all the equivalent representations of a string diagram modulo Frobenius into a single object, thus easing the complexity of matching. The idea of the isomorphism is that boxes in a string diagram are represented as hyperedges, and wires as nodes. The use of cospans $I_l \tr{f} G \tl{g} I_r$ is essential: the carrier $G$ encodes the string diagram itself, whereas $I_l$ and $I_r$ are discrete hypergraphs (sets of nodes) that indicate through $f$ and $g$ which nodes of $G$ are dangling wires on the left and on the right of the corresponding string diagram.

As the name suggests, the close relationship between hypergraph categories and hypergraph structures was clear to previous authors \cite{KissingerHypergraph}, as well as the use of cospans to mimic interfaces \cite{Ehrig2004,Rosebrugh2005}. Our characterisation combines existing approaches in a way that best suits the application to rewriting. The main generalisation is characterising hypergraph categories that are freely generated by multi-sorted instead of single-sorted signatures. 

\item The second step (Section \ref{sec:rewriting}) is to exploit the combinatorial interpretation to realise string diagram rewriting as rewriting of hypergraphs. The fact that cospans of hypergraphs form an adhesive category \cite{Lack2005} gives an off-the-shelf theory of double-pushout (DPO) rewriting. We show that rewriting modulo Frobenius and DPO rewriting of hypergraphs are essentially the same thing. The problem of matching in an hypergraph category \eqref{eq:matchingproblem} is reduced to finding an hypergraph homomorphism.

\item As last contribution (Section \ref{sec:confluence}), we show that confluence for terminating rewriting systems in hypergraph categories is decidable, as it is reducible to a computable critical pair analysis. This well-known property of term rewriting becomes a non-trivial question when dealing with two-dimensional entities. For instance, in the aforementioned polygraph approach~\cite{Burroni1993}, where critical pairs are considered in the string diagrammatic syntax rather than in a graph model, even a finite set of rewriting rules may yield infinitely many critical pairs~\cite{MimramFix}. In the context of ordinary DPO graph rewriting decidability also fails~\cite{Plump1993} unless further conditions are imposed, such as requiring that all critical pairs satisfy a syntactic condition called coverability~\cite{Plump10} or that they are joinable in a stricter sense~\cite{EhrigHPP04}. We establish our decidability result within the framework of recent work \cite{BGKSZ-esop17} that studies confluence for DPO hypergraph rewriting \emph{with interfaces}. Not only this variant enjoys decidability without further restrictions on critical pairs, but is precisely tailored for the interpretation of ``syntactic'' rewriting from hypergraph categories. We refer to \cite{BGKSZ-esop17} for a more extensive discussion of how the interface approach compares to others in the DPO rewriting literature. 
\end{itemize}

Rewriting modulo Frobenius structure has been studied along the same lines in \cite{BGKSZ-lics16,BGKSZ-esop17}. These recent works by the author and collaborators serve as a roadmap for this paper: the aim here is to verify that such results generalise to multi-sorted algebraic theories, for which the freely generated category has a Frobenius structure on each sort. In light of \cite{BGKSZ-lics16,BGKSZ-esop17}, the way this generalisation unfolds is not particularly surprising, as we are essentially able to lift the same proof techniques from a single to multiple sorts. However, we believe that the redaction of a reference paper for these results is timely. Firstly, it is justified by the renewed interest for hypergraph categories, witnessed by several recent applications, especially to circuit theory \cite{Fon16} and to natural language semantics \cite{MG17,KartsaklisSPC14}: using the theory developed in \cite{BGKSZ-lics16,BGKSZ-esop17} is going to require the full generality of the multi-sorted case. Secondly, another justification comes from axiomatic approaches to various families of systems (concurrent ~\cite{Bruni2006}, quantum~\cite{Coecke2008}, dynamical~\cite{Bonchi2014b,BaezErbele-CategoriesInControl}) in which the equational theory axiomatising system behaviour includes two or more Frobenius algebras. When it comes to rewriting, the approach introduced in \cite{BGKSZ-lics16} only allows to absorb \emph{one} Frobenius structure in the combinatorial model. In this paper, we show how additional Frobenius structures can be also absorbed in the same manner\footnote{The fact that multiple Frobenius structures are on the same object (like in the aforementioned theories) or on different objects (like in this work) of a category may be overcome with the addition of `switch' operations from one object to the others, as we are going to show in a paper in preparation.}, thus reducing the complexity of the aforementioned axiomatisations and simplifying the task of studying normal forms, confluence and termination.


\paragraph{Notation.}
In a category $\catC$ with coproducts, $\copair{h_1,h_2} \: X+Y \to Z$ is the copairing of $h_1 \: X \to Z$ and $h_2 \: Y \to Z$, defined by universal property of $+$. Also, $f \poi g \colon a \to c$ is the composition of arrows $f \colon a \to b,\, g\colon b \to c$. We sometimes write $a\tr{f}b$ or $b \tl{f} a$ for $f \: a \to b$, or also $\tr{f}$ and $\tl{f}$ if object names are immaterial for the context. We write $\tns$ for the monoidal product in a monoidal category.


\section{Props and Hypergraph Categories}\label{sec:background}



We are going to study hypergraph categories freely generated by a signature of operations. The following is the notion appropriate to the monoidal context.

\begin{definition} A \emph{monoidal theory} is a tuple $(\Sigma, \col)$ of a \emph{signature} $\Sigma$ and a finite set $\col$ of \emph{colours}. Elements of $\Sigma$ are \emph{operations} $o \: w \to v$ with a certain \emph{arity} $w$ and \emph{coarity} $v$, where $w, v \in \col^{\star}$.  
\end{definition}

Generic theories are typically triples, allowing also for a set of equations on $\Sigma$-terms. We do not need that level of generality here: equations will be treated differently, as rewriting rules (unless they are structural, like the equations of symmetric monoidal categories or of hypergraph categories, see below).


Towards hypergraph categories, it is instrumental to describe first the free \emph{symmetric} monoidal category generated by a theory $(\Sigma,\col)$, which is called a $\col$-coloured prop~\cite{HackneyColouredPROPs15} (\textbf{pro}\-duct and \textbf{p}ermutation category). This works in analogy with the single-sorted case $\col = \{ c \}$, in which monoidal theories act as presentations for ($\{c\}$-coloured) props~\cite{Lack2004a}.

\begin{definition}\label{def:prop} Let $\col$ be a finite set. A \emph{$\col$-coloured prop} is a symmetric monoidal category (SMC) where the set of objects is $C^{\star}$ and the monoidal product $\tns$ on objects is word concatenation. $\col$-coloured props form a category $\PROP{\col}$ with morphisms the identity-on-objects symmetric monoidal functors.

Given a monoidal theory $(\Sigma, \col)$, one can freely construct a prop $\syntax{\Sigma , \col}$ with arrows the $\Sigma$-terms quotiented by the laws of symmetric monoidal categories. $\Sigma$-terms are freely obtained by combining operations in $\Sigma$, a \emph{unit} $\id \colon c\to c$ for each $c \in \col$ and a \emph{symmetry} $\sigma_{c,d} \colon cd\to dc$ for each $c,d\in \col$, by sequential ($;$) and parallel ($\tns$) composition. That means, given terms $a \colon w_1 \to w_2$, $b \colon w_2\to w_3$, $a' \colon v_1\to v_2$, one constructs new terms $a \poi b \colon w_1\to w_3$ and $a \tns a' \colon w_1 v_1 \to w_2 v_2$.
\end{definition}

We shall adopt the graphical notation of string diagrams~\cite{Selinger2009} for the arrows of $\syntax{\Sigma,\col}$. An arrow $a \: w_1 \to w_2$ is pictured as $\cgr{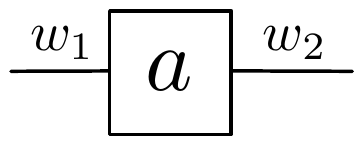}$. Compositions via $\poi$ and $\tns$ are drawn respectively as horizontal and vertical juxtaposition, that means, $a \poi b$ is drawn $\cgr[height=.6cm]{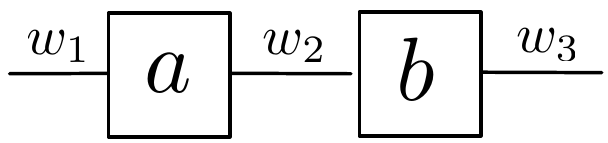}$ and $a \tns a'$ is drawn $\cgr[height=.8cm]{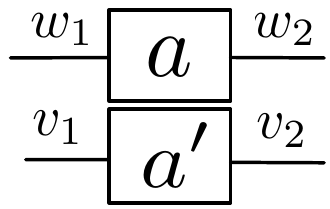}$. There are specific diagrams for the symmetric monoidal structure, namely $\cgr{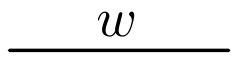}$ for the identity $\id_w \: w \to w$ and $\cgr{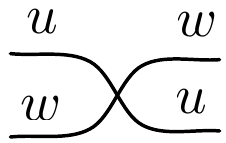}$ for the symmetry $\sigma_{w,u} \: wu \to uw$, for $w, u \in \col^{\star}$. These are definable from the basic identities and symmetries for colours in $\col$ using the pasting rules for $\poi$ and $\tns$.

\begin{example}\label{ex:perm}
The initial object in $\PROP{\col}$ is the $\col$-coloured prop $\perm{\col}$ whose arrows $w \to v$ are permutations of $w$ into $v$ (thus arrows exist only when the word $v$ is an anagram of the word $w$). $\perm{\col}$ is freely generated by the monoidal theory $(\emptyset,\col)$.
\end{example}

\begin{example} 
For $\col' \subseteq \col$, the $\col$-coloured prop $\frob{\col'}$ of \emph{separable Frobenius $\col'$-monoids} is freely generated by the monoidal theory $(\Sigmafrob{\col'}, \col)$, where $\Sigmafrob{\col'} = \{ \ \cgr{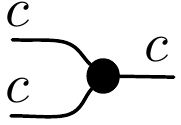} \ , \ \cgr{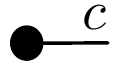} \ , \ \cgr{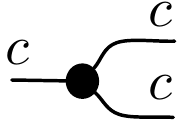} \ , \ \cgr{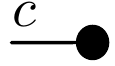} \mid c \in \col' \}$, and quotiented by equations, for each $c \in \col'$, stating that $\cgr{Bmultc.pdf}$ and $\cgr{Bunitc.pdf}$ form a commutative monoid \eqref{eq:multc}, that $\cgr{Bcomultc.pdf}$ and $\cgr{Bcounitc.pdf}$ form a commutative comonoid \eqref{eq:comultc}, and that these interact according to the Frobenius law and the separability law \eqref{eq:frobc}.
\begin{gather}\label{eq:multc} 
\cgr{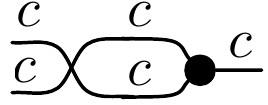} = \cgr{Bmultc.pdf} \qquad \qquad \cgr{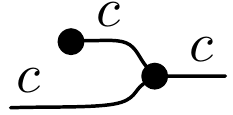} = \cgr{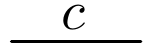} \qquad \qquad \cgr{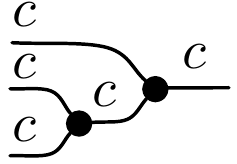} = \cgr{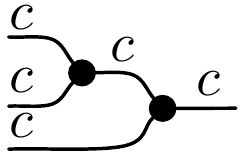} 
\\ 
\label{eq:comultc} 
\cgr{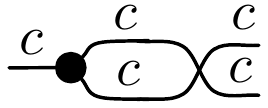} = \cgr{Bcomultc.pdf} \qquad \qquad \cgr{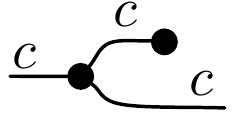} = \cgr{idc.pdf} \qquad \qquad \cgr{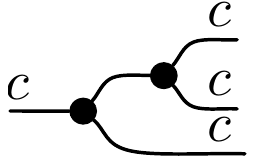} = \cgr{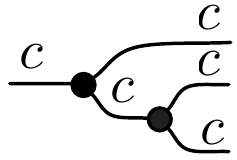}
\\
\label{eq:frobc} 
\cgr{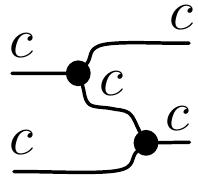} = \cgr{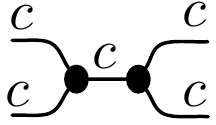} = \cgr{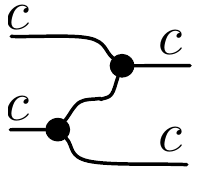}
\qquad \qquad 
\cgr{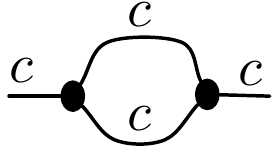} = \cgr{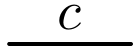} 
\end{gather}
\end{example}
As mentioned, when $\col$ is a singleton, $\col$-coloured props are just called props, i.e. SMCs with objects the natural numbers where the monoidal product is addition on objects. For later use it is convenient to record the following result about the single-sorted case. It involves the prop $\Cospan{\FINSET}$ whose arrows $n_1 \to n_2$ are cospans $\ord{n_1} \tr{f} \ord{n_3} \tl{g} \ord{n_2}$ of functions between ordinals $\ord{n} \df \{0,\dots,n-1\}$.

\begin{proposition}[\cite{Bruni01somealgebraic,Lack2004a}]\label{prop:cospanfrob} There is an isomorphism of $\{c\}$-coloured props between $\frob{\{c\}}$ and $\Cospan{\FINSET}$. It is defined by the following mapping on the $\Sigmafrob{\{c\}}$-operations.
\begin{gather*}
\cgr{Bmultc.pdf} \quad \mapsto \quad \cgr{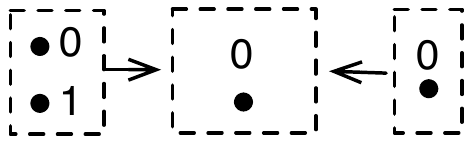} \qquad 
\cgr{Bcomultc.pdf} \quad \mapsto\qquad \cgr{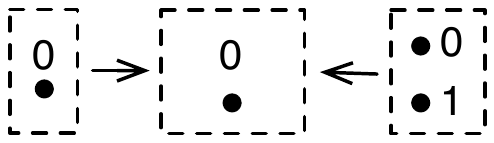} \\
\cgr{Bunitc.pdf} \quad \mapsto \quad \cgr{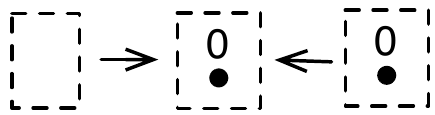} \qquad 
\cgr{Bcounitc.pdf} \quad \mapsto \quad \cgr{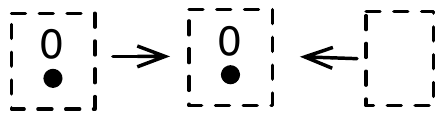}
\end{gather*}
\end{proposition}

As outlined in the introduction, we are interested in studying non-structural equations as rewriting rules. We now define the appropriate notion of rewriting for arrows in a prop. We call it ``syntactic'' to emphasise that matching happens when the left-hand side of a rule is a sub-term.

\begin{definition}[Syntactic rewriting in a prop] \label{defn:rewprop}
A \emph{rewriting rule} in a $\col$-coloured prop $\catC$ is a pair of morphisms $l, r \: v_1 \to v_2$ in $\catC$, for which we use the notation $\rrule{l}{r} \: v_1 \to v_2$. A \emph{rewriting system} $\RS$ is a finite set of rewriting rules. Given $a,b \: w_1 \to w_2$ in $\catC$, we say that $a$ rewrites into $b$ via $\RS$, notation $a \Rew{\RS} b$, if there are $a_1$ and $a_2$ yielding the following decompositions in $\catC$, where $\rrule{l}{r} \: v_1 \to v_2$ is in $\RS$.
\begin{equation}\label{eq:rewpropmatch} \cgr{diagA.pdf} = \cgr{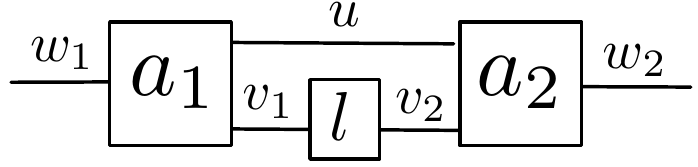} \qquad \qquad \cgr{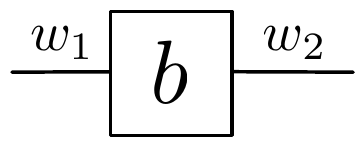} = \cgr{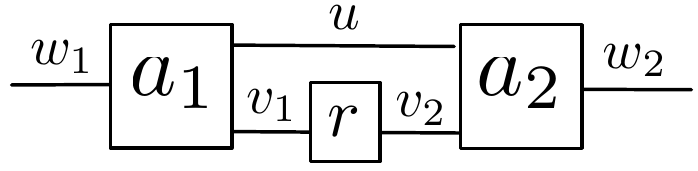} \end{equation} 
\end{definition}

We are going to study syntactic rewriting in free hypergraph categories, which we now introduce together with their properties.

\begin{definition} An \emph{hypergraph category} is an SMC $\catA$ where each object $x \in\catA$ has a separable Frobenius structure, i.e., maps
$\cgr{multx.pdf}$, $\cgr{unitx.pdf}$, $\cgr{comultx.pdf}$ and $\cgr{counitx.pdf}$ forming a commutative monoid, a commutative comonoid and satisfying equations as in \eqref{eq:frobx} for each $x \in \catA$. Moreover, the Frobenius structure must be compatible with the monoidal product:
$$\cgr{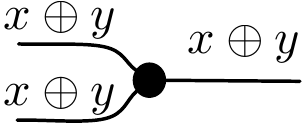}\ =\ \cgr{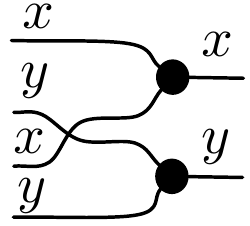} \qquad \qquad
 \cgr{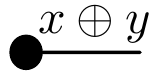} \ =\ \cgr{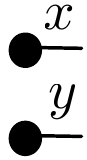} \qquad \qquad 
\cgr{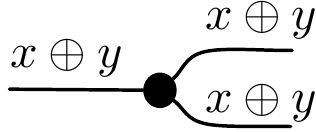}\ =\ \cgr{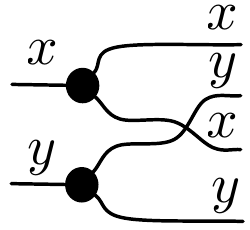} \qquad \qquad 
\cgr{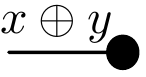} \ =\ \cgr{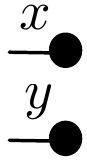}$$

The free hypergraph category over $(\Sigma, \col)$, notation $\freehyp{\Sigma , \col}$, is the free $\col$-coloured prop on $(\Sigma \uplus \Sigmafrob{\col},\col)$ quotiented by equations \eqref{eq:multc}, \eqref{eq:comultc} and \eqref{eq:frobc} for each $c \in \col$.
\end{definition}
 
 \medskip

Observe that the free construction of $\freehyp{\Sigma,\col}$ indeed creates a Frobenius structure for each object $w \in \col^{\star}$ of the category, canonically defined in terms of the one on colours. For instance, for $w = c_1 \dots c_n$:
\[ \cgr{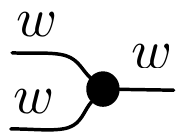} \ \df \ \cgr{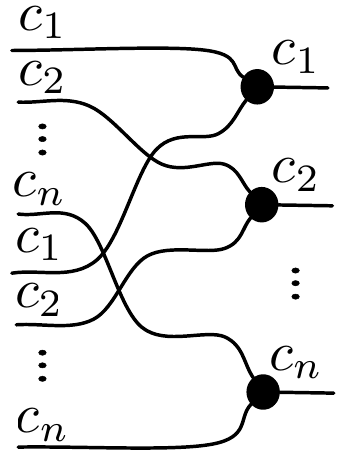} \qquad\qquad
\cgr{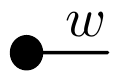}  \ \df \  \cgr{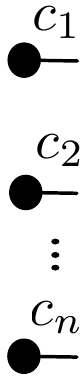} \qquad\qquad
\cgr{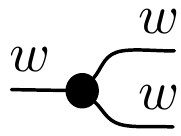}  \ \df \  \cgr{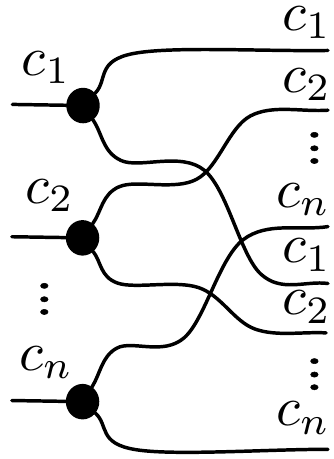} \qquad\qquad
\cgr{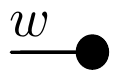}  \ \df \  \cgr{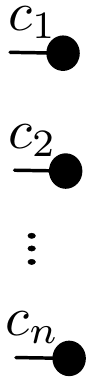} \]

\begin{example}\label{ex:bipartite}
Fix a set $\col$ of colours with just two elements, noted $\sr$ and $\sg$, and a signature $\Sigma$ consisting of two ``colour switch'' operations, $\cgr{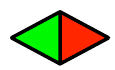} \: \sg \to \sr$ and $\cgr{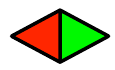} \: \sr \to \sg$. We may construct the free hypergraph category $\freehyp{\Sigma,\{\sr,\sg\}}$ over $(\Sigma,\{\sr,\sg\})$. Here is an example of a string diagram in this category, where we use the more suggestive convention of colouring wires instead of labelling them with objects $\sr$ and $\sg$.
\begin{equation}\label{eq:bipartite}
\cgr{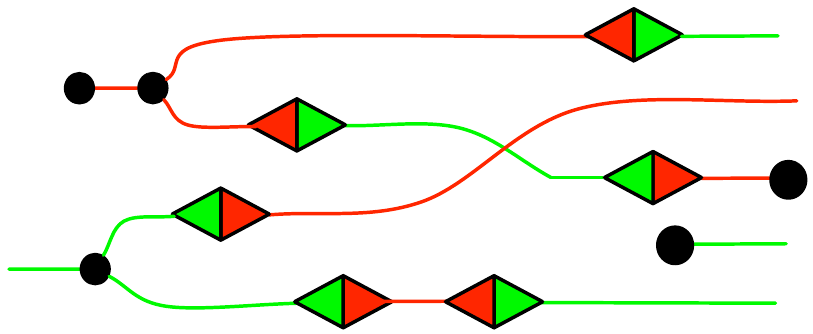}
\end{equation}
We claim that $\freehyp{\Sigma,\{\sr,\sg\}}$ is the same as the category of finite directed \emph{bipartite graphs} (with interfaces). This will become clear in Example \ref{ex:bipartite2}, after the characterisation provided by Corollary \ref{cor:freehyp=fterm}.\end{example}


We now observe that the free hypergraph category $\freehyp{\Sigma,\col}$ can be seen as a coproduct in $\PROP{\col}$. This will be useful in order to separate the component arising from $\Sigma$ from the built-in Frobenius structure.

\begin{proposition}\label{th:syncharacterisation} There is an isomorphism $\syntax{\Sigma , \col} + \frob{C} \cong \freehyp{\Sigma,\col}$ of $\col$-coloured props. It extends to an isomorphism of hypergraph categories: thus $\syntax{\Sigma , \col} + \frob{C}$ is the free hypergraph category on $(\Sigma,\col)$. \end{proposition}
\begin{proof}
The first part follows from how coproducts are computed in $\PROP{\col}$. As $\syntax{\Sigma,\col}$ is presented by $(\Sigma, \col)$ and $\frob{C}$ by $(\Sigmafrob{\col},\col)$ quotiented by \eqref{eq:multc}-\eqref{eq:frobc}, then $\syntax{\Sigma , \col} + \frob{C}$ is presented by $(\Sigma \uplus \Sigmafrob{\col}, \col)$ quotiented by \eqref{eq:multc}-\eqref{eq:frobc}: this is precisely the definition of $\freehyp{\Sigma,\col}$. The second part holds because the isomorphism maps the Frobenius structure on $c \in \col$ in $\syntax{\Sigma , \col} + \frob{C}$ to the Frobenius structure on $c \in \col$ in~$\freehyp{\Sigma,\col}$.
\end{proof}

The free hypergraph category has a universal property (of a pushout) in $\SMC$ too. Details are in the proof of Corollary \ref{cor:faithful}, which uses this observation. We conclude by recalling that the Frobenius monoids in an hypergraph category defines a canonical compact closed structure. This also justifies the terminology ``well-supported compact closed categories'' originally used for hypergraph categories \cite{Carboni1987}. 

\begin{proposition}[\cite{Carboni1987}]\label{prop:compactclosed}
Hypergraph categories are (self-dual) compact closed.
\end{proposition}
\begin{proof}
It is useful to report how the compact closed structure is actually defined. For an object $x$ of an hypergraph category $\catA$, define
$\cgr{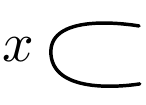}$ as $\cgr{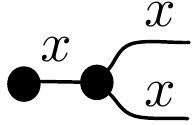}$ and $\cgr{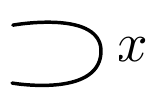}$ as $\cgr{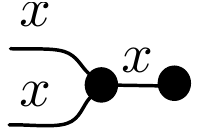}$. The Frobenius equation \eqref{eq:frobc} implies the equation for compact closure:
$$\cgr{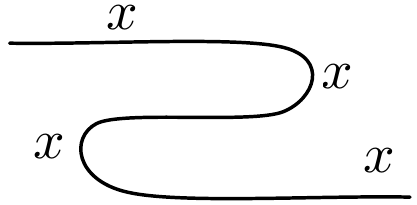} = \cgr{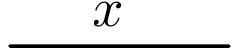} = \cgr{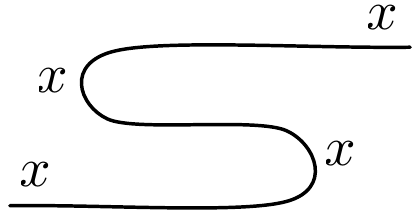}.$$
The dual $\cgr{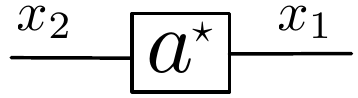}$ of a morphism $\cgr{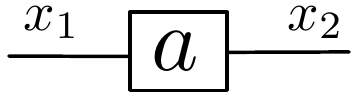}$ is defined as $\cgr{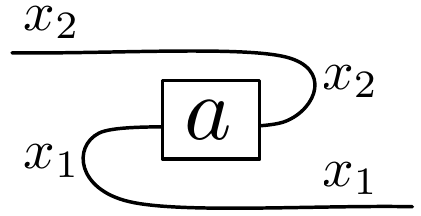}$.

\end{proof}

\section{The Combinatorial Interpretation}\label{sec:combinatorialmodel}

According to Definition \ref{defn:rewprop}, syntactic rewriting in free hypergraph categories happens modulo Frobenius structure. The goal of this section is to give a combinatorial description of the free hypergraph category, so that a more concrete account of the associated rewriting becomes available.

Fix a monoidal theory $(\Sigma, \col)$. We shall work with finite directed hypergraphs, whose hyperedges are labelled in $\Sigma$ and nodes are labelled in $\col$. We shall visualise hypergraphs as follows: $\cgr{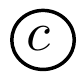}$ is a node labeled

\noindent \begin{minipage}{.7\textwidth}
with $c \in \col$ and $\cgr{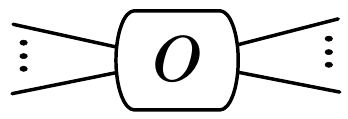}$ is an hyperedge labeled with $o \in \Sigma$, with ordered tentacles attached to the left boundary linking to sources and the ones on the right linking to targets. An example is on the right, with $\col = \{ c_1, c_2\}$ and $\Sigma = \{ o_1 \: c_1 \to \epsilon , o_2 \: c_1 c_2 \to c_1 c_1 \}$.
\end{minipage}
\begin{minipage}{.2\textwidth}
\vspace{-.5cm}\[\quad \cgr[width=4cm]{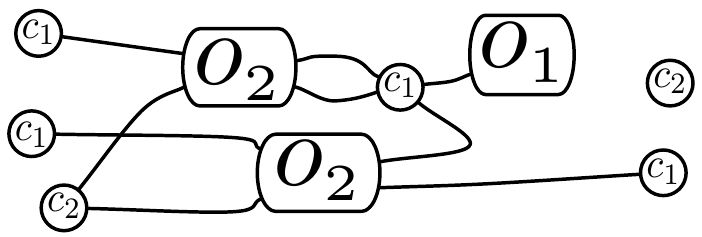}\]
\end{minipage}

We now organise these structures into a category. First, consider the SMC $\Hyp{}$ with objects the finite directed (unlabelled) hypergraphs and homomorphisms between them. The monoidal theory $(\Sigma, \col)$ itself can be seen as an object of $\Hyp{}$. For instance, $\Sigma$ and $\col$ as above yield the unlabelled hypergraph on the 

\noindent \begin{minipage}{.3\textwidth}
\[\cgr[width=4cm]{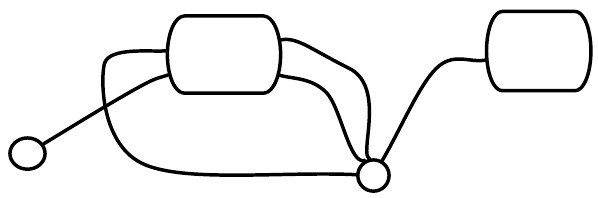}\]
\end{minipage}
\begin{minipage}{.7\textwidth}
left (where we ``call'' $o_2$ the leftmost and $o_1$ the rightmost hyperedge, and $c_2$ the leftmost and $c_1$ the rightmost node). The desired labelling is given formally by working in the slice category $\Hyp{} \setminus (\Sigma,\col)$, for which we shall use notation $\Hyp{\Sigma, \col}$. This definition ensures that a $\Sigma$-operation $o \: w \to v$ labels an hyperedge only when the label of its input (respectively, output) nodes forms the word $w$ ($v$).
\end{minipage}

The next step is to add interfaces, formally modelled by cospans. Fix a set $\{x_i\}_{i \in \N}$ totally ordered by $\N$. Define $\FTerm{\Sigma , \col}$ as the restriction of the category of cospans in $\Hyp{\Sigma, \col}$ to objects the discrete hypergraphs (i.e. no hyperedges) whose set of nodes is isomorphic to an initial segment of $\{x_i\}_{i \in \N}$. This restriction ensures that $\FTerm{\Sigma , \col}$ is a $\col$-coloured prop. Indeed, an object of $\FTerm{\Sigma , \col}$ can be identified with a natural number $k$ together with a labelling function $\{0,\dots,k-1\} \to \col$, which is the same as a word in $\col^*$. The notation $\FTerm{\Sigma , \col}$ stands for ``Frobenius termgraphs'', following the terminology introduced for the single-sorted case \cite{BGKSZ-lics16}. This name will be justified by the characterisation of Theorem~\ref{th:characterisation} below.

We now define the two components of the functor that is going to interpret the syntactic definition of an hypergraph category as a combinatorial structure. The key to the approach is the second definition below, which essentially tells that the combinatorial model is able to absorb all the complexity of Frobenius structure simply in terms of nodes.

\begin{definition} We define a $\col$-coloured prop morphisms $\allTosem{\cdot} \: \: \syntax{\Sigma , \col} + \frob{\col} \to \FTerm{\Sigma,\col}$ as the copairing of the following functors:
\begin{itemize}
\item $\synTosem{\cdot} \: \syntax{\Sigma , \col} \to \FTerm{\Sigma,\col}$ is defined by the following mapping on operations $o \: c_1c_2\dots c_m \to b_1b_2 \dots b_m$ of $\Sigma$ (where $c_i, b_i \in \col$).
$$\cgr{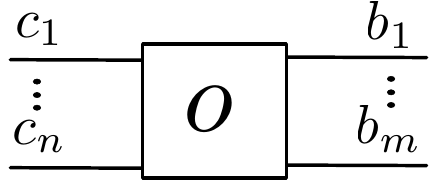} \quad \xmapsto{\synTosem{\cdot}} \quad \cgr{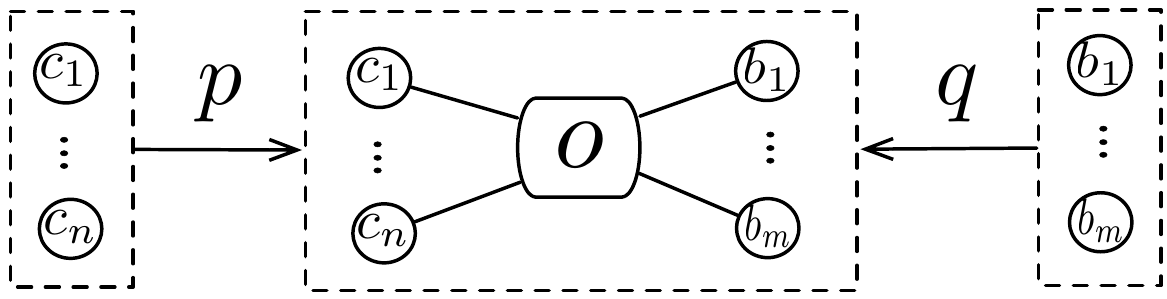}$$
The definition of hypergraph homomorphisms $p$ and $q$ is fixed by colour preservation.
\item $\frobTosem{\cdot} \: \frob{C} \to \FTerm{\Sigma,\col}$ is defined by the following mapping on the $\Sigmafrob{\col}$-operations.
\begin{equation}\label{eq:frobTosem}
\begin{aligned}
\cgr{Bunitc.pdf} \quad \xmapsto{\frobTosem{\cdot}} \quad \cgr{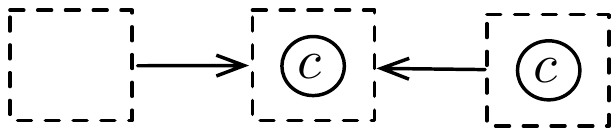} \qquad \cgr{Bcounitc.pdf} \quad \xmapsto{\frobTosem{\cdot}} \quad \cgr{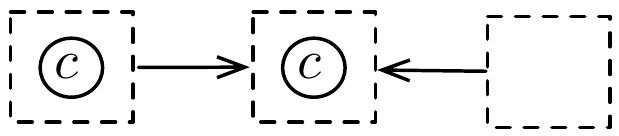} \\
\cgr{Bmultc.pdf} \quad \xmapsto{\frobTosem{\cdot}} \quad \cgr{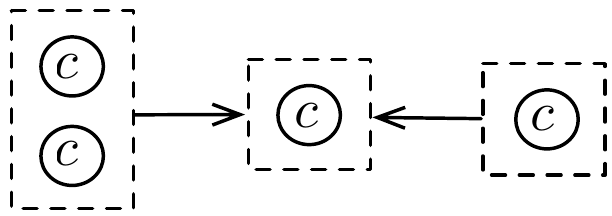} \qquad \cgr{Bcomultc.pdf} \quad \xmapsto{\frobTosem{\cdot}} \quad \cgr{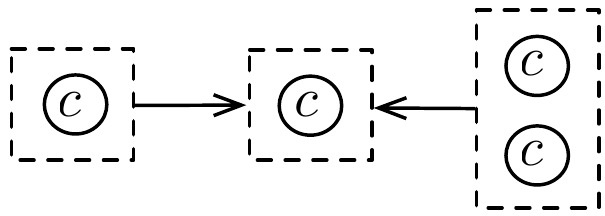}
\end{aligned}
\end{equation}
Also here the definition of the hypergraph homomorphisms is predetermined.
\end{itemize}
\end{definition}

\medskip

We now have all the ingredients to state our characterisation theorem.

\begin{theorem}\label{th:characterisation} $\allTosem{\cdot} \: \syntax{\Sigma , \col} + \frob{C} \to \FTerm{\Sigma , \col}$ is an isomorphism of $\col$-coloured props. \end{theorem}

The proof of the theorem will be postponed to the end of the section. Let us first observe two interesting consequences.

\begin{corollary}\label{cor:freehyp=fterm} $\FTerm{\Sigma , \col}$ is the free hypergraph category on $(\Sigma, \col)$, i.e. $\freehyp{\Sigma,\col} \cong \FTerm{\Sigma , \col}$. \end{corollary}
\begin{proof} $\frobTosem{\cdot} \: \frob{C} \to \FTerm{\Sigma , \col}$ defines an hypergraph category structure on $\FTerm{\Sigma , \col}$ and the isomorphism of Theorem \ref{th:characterisation} extends to one of hypergraph categories. Then the result follows by Proposition~\ref{th:syncharacterisation} and Theorem \ref{th:characterisation}. \end{proof}

The next corollary states that there is no `information loss' in passing from the free symmetric monoidal category to the free hypergraph category on $(\Sigma,\col)$.
\begin{corollary} \label{cor:faithful} $\synTosem{\cdot}  \: \syntax{\Sigma , \col} \to \FTerm{\Sigma,\col}$ is faithful. \end{corollary}
\begin{proof}
We use that, just as for props \cite[Prop.~2.8]{ZanasiThesis}, coproducts of $\col$-coloured props can be computed as certain pushouts in the category $\SMC$ of small SMCs. In particular, $\syntax{\Sigma , \col} + \frob{\col}$ in $\PROP{\col}$ arises as

 \begin{equation}\label{eq:pushout}
 \vcenter{
\xymatrix@R=20pt@C=40pt{
 \ar[d]_{!_1} \perm{\col}  \ar[r]^-{!_2} \drcorner & \frob{C} \ar[d]^-{\frobTosem{\cdot}}\\
\ar[r]_-{\synTosem{\cdot}} \syntax{\Sigma , \col} & {\syntax{\Sigma , \col} + \frob{\col} \cong \FTerm{\Sigma,\col}}
}
}
\end{equation}
in $\SMC$, where the maps $!_1$ and $!_2$ are given by initiality of $\perm{\col}$ in $\PROP{\col}$ (see Example \ref{ex:perm}). Intuitively, in \eqref{eq:pushout} $\syntax{\Sigma , \col} + \frob{\col}$ is built as the ``disjoint union'' of $\syntax{\Sigma , \col}$ and $\frob{\col}$ where one identifies the set of objects $\col^{\star}$ and the associated symmetric monoidal structure (the ``contribution'' of $\perm{\col}$).

Now, in order to prove that $\synTosem{\cdot}$ is faithful, we can use a result~\cite[Th.~3.3]{MacDonald2009} about amalgamation in $\CAT$ (which transfers to $\SMC$). As all the functors in \eqref{eq:pushout} are identity-on-objects and $!_1$, $!_2$ are faithful, it just requires to show that $!_1$ and $!_2$ satisfy the so-called 3-for-2 property: for $!_1$, this means that, given $h= f \poi g$ in $\syntax{\Sigma , \col}$, if any two of $f,g,h$ are in the image of $!_1$, then so is the third. This trivially holds as every arrow of $\perm{\col}$ is an isomorphism. The argument for $!_2$ is identical. 
\end{proof}

\begin{example}\label{ex:bipartite2}
We come back to the free hypergraph category $\freehyp{\Sigma,\{\sr,\sg\}}$ introduced in Example \ref{ex:bipartite2}. By Corollary \ref{cor:freehyp=fterm}, $\freehyp{\Sigma,\{\sr,\sg\}} \cong \FTerm{\Sigma,\{\sr,\sg\}}$. In hypergraphs of $\FTerm{\Sigma,\{\sr,\sg\}}$, hyperedges correspond to switches $\cgr{bgswitch.pdf}$ or $\cgr{gbswitch.pdf}$, thus they are in fact edges (one input and one output node) and we may as well avoid drawing them in the graphical representation. Since they connect any two nodes only when these have a different colour, what we obtain are finite directed bipartite graphs. For instance, reprising \eqref{eq:bipartite}:
\[ \cgr{exbipartite.pdf} \qquad \xmapsto{\allTosem{\cdot}}\qquad \cgr{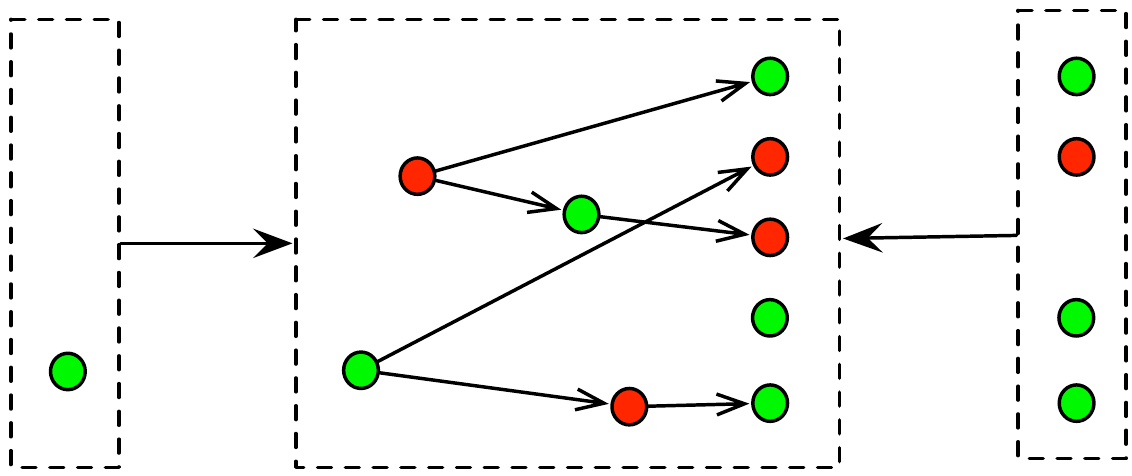}\]
This example shows that Theorem \ref{th:characterisation} not only provides a combinatorial representation for algebraic structures, but conversely it is also instrumental in deriving an algebraic presentation for well-known graph-theoretic models.
\end{example}

We now give a proof of the characterisation theorem.

\begin{proof}[Proof of Theorem \ref{th:characterisation}]
As a preparatory step, we observe that $\frob{\col}$ itself can be decomposed as a coproduct in $\PROP{\col}$, namely $\Sigma_{c \in \col} \frob{\{c\}}$. Let us suppose for simplicity that $\col = \{c,d\}$, so that $\frob{\col} = \frob{\{c\}} + \frob{\{d\}}$. It will be apparent how the argument generalises. 

The prop morphism $\frobTosem{\cdot} \: \frob{\col} \to \FTerm{\Sigma, \col}$ is analogously decomposable as the copairing $\copair{\frobTosem{\cdot}_{c},\frobTosem{\cdot}_{d}} \: \frob{\{c\}} + \frob{\{d\}} \to \FTerm{\Sigma,\col}$ of prop morphisms $\frobTosem{\cdot}_{c} \: \frob{\{c\}} \to \FTerm{\Sigma,\col}$ and $\frobTosem{\cdot}_{d} \: \frob{\{d\}} \to \FTerm{\Sigma,\col}$, defined by restricting the clauses \eqref{eq:frobTosem} to the associated colour, either $c$ or $d$.

We have thus reduced the statement to verify that $\FTerm{\Sigma , \col}$ satisfies the universal property of the coproduct $\syntax{\Sigma,\col} + \frob{\{c\}} + \frob{\{ d \}}$ in $\PROP{\col}$. 
\begin{equation} \label{eq:univcoproductgen}
\vcenter{\xymatrix@R=5pt@C=45pt{
&&&& \ar[dll]_-{\frobTosem{\cdot}_{d}} \ar@/_8pt/[dddll]^<<<<<<{\beta_d} {\frob{\{d\}}} \\
{\syntax{\Sigma,\col}} \ar[ddrr]_\alpha \ar[rr]^-{\synTosem{\cdot}} && 
{\FTerm{\Sigma,\col}} \ar@{.>}[dd]^\gamma \\
&&&& \ar[ull]_<<<<<<{\frobTosem{\cdot}_{c}} {\frob{\{c\}}} \ar[dll]^{\beta_c} & \\
&& {\catA}
}}
\end{equation}
 Given $\alpha$, $\beta_c$, $\beta_d$ and a $\col$-coloured prop $\catA$ as in \eqref{eq:univcoproductgen}, we need to show the existence of a unique $\gamma$ making the diagram commute. Now, because all morphisms in \eqref{eq:univcoproductgen} are identity-on-objects, it suffices to show that any arrow of $\FTerm{\Sigma , \col}$ can be decomposed in an essentially unique way into an expression where all the basic constituents lie in the image of $\synTosem{\cdot}$, $\frobTosem{\cdot}_{c}$ or $\frobTosem{\cdot}_{d}$. 
 
To this aim, fix a cospan $w \tr{f} G \tl{g} v$ in $\FTerm{\Sigma , \col}$, where $G$ has set of nodes $N$, set of hyperedges $E$ and labelling functions $\psi \: N \to \col$ and $\chi \: E \to \Sigma$. We pick an order $e_1,\dots,e_j$ on the hyperedges in $E$ and one $n_1,\dots,n_k$ on the nodes in $N$. Let $\tilde{w} \tr{i} \tilde{E} \tl{o} \tilde{v}$ be the cospan defined as $\bigoplus_{1 \leq i \leq j} \synTosem{\chi(e_i)}$. Intuitively, $\tilde{E}$ piles up all the hyperdges of $G$, but disconnected from each other. $\tilde{w}$ and $\tilde{m}$ are the word concatenations of all the inputs, respectively outputs of these hyperedges. 

Similarly, we pile up all the (labeled) nodes in $N$, by forming the word $w_N \in \col^{\star}$ as $\bigoplus_{1 \leq i \leq k} \psi(n_i)$. There are obvious functions from $w$, $v$, $\tilde{w}$ and $\tilde{v}$ to $w_N$, mapping labelled nodes to their occurrence in $w_N$. All this information is now gathered in the following composition of cospans \footnote{We admit a certain degree of sloppiness in writing $w_N$ both for an object of $\FTerm{\Sigma,\col}$ and for the carrier of a cospan. For $w_N = \{n_1, \dots, n_k\}$, these are isomorphic descriptions of the same data: in the first case it is treated as a word in $\col^{\star}$, in the second as a set of nodes with a labelling function $\chi \: \{ n_1, \dots, n_k \} \to \col$. }
\begin{equation}\label{eq:decomp}
(w \tr{f}  w_N \tl{\copair{id, j}} w_N \tns \tilde{w}) \ \poi \ (w_N \tns \tilde{w} \tr{id \tns i} w_N \tns \tilde{E} \tl{id \tns o} w_N \tns \tilde{v}) \ \poi \ (w_N \tns \tilde{v} \tr{\copair{id, p}} w_N \tl{g} v)
\end{equation}
Copairing maps $\copair{id, j}$ and $\copair{id, p}$ are well-defined as $\tns$ is also a coproduct in $\Hyp{\Sigma,\col}$. One can compute that the result of composing \eqref{eq:decomp} (by pushout) is indeed isomorphic to $w \tr{f} g \tl{g} v$.

Towards a definition of $\gamma$, we need to check that every component of \eqref{eq:decomp} is in the image of either $\synTosem{\cdot}$ or $\frobTosem{\cdot} = \copair{\frobTosem{\cdot}_{c},\frobTosem{\cdot}_{d}}$. The middle cospan is clearly in the image of $\synTosem{\cdot}$, as it is the monoidal product of the identity cospan $w_N \tr{id} w_N \tl{id} w_N$ with cospans in the image of some $o \in \Sigma$. Next, we want to check that the two outmost cospans are in the image of $\copair{\frobTosem{\cdot}_{c},\frobTosem{\cdot}_{d}}$. To this aim, we show the following claim.

\begin{claim} Any arrow $u_1 \tr{h} u_3 \tl{q} u_2 $ of $\FTerm{\Sigma,\col}$ with $u_1$, $u_2$, $u_3$ discrete is in the image of $\copair{\frobTosem{\cdot}_{c},\frobTosem{\cdot}_{d}}$. 
\end{claim}
\begin{proof} First, find permutations $\pi_1 \: u_1 \to c^k d^z$, $\pi_2 \: u_2 \to c^m d^n$ and $\pi_3 \: u_3 \to c^l d^r$ factorising words $u_1$, $u_2$ and $u_3$ respectively as $c$s followed by $d$s. We can then define restrictions of $h$ and of $q$ to the $c$-segment or the $d$-segment of their domain: this gives functions $h_c \: c^k \to c^l$, $h_d \: d^z \to d^r$, $q_c \: c^m \to c^l$ and $q_d \: d^n \to d^r$. Observe that the codomain is restricted too, as $h$ and $q$ are colour-preserving maps. Putting these data together we can decompose $u_1 \tr{h} u_3 \tl{q} u_2 $ as follows.
\begin{equation}\label{eq:decompFrob}
(u_1 \tr{id}  w_1 \tl{\pi_1} c^k \tns d^z) \ \poi \ (c^k \tns d^z \tr{h_c \tns h_d} c^l \tns d^r \tl{q_c \tns q_d} c^m \tns d^n) \ \poi \ (c^m \tns d^n \tr{\pi_2} u_2 \tl{id} u_2)
\end{equation}
It is now useful to recall Proposition \ref{prop:cospanfrob}. Observe that the bijection given therein between Frobenius structure and cospans in $\FINSET$ is defined by the same clauses \eqref{eq:frobTosem} as $\frobTosem{\cdot} = \copair{\frobTosem{\cdot}_{c},\frobTosem{\cdot}_{d}}$, modulo the labelling of set elements all with $c$ or with $d$. It follows that the cospan $c^k \tr{h_c} c^l \tl{q_c} c^m$ is in the image of $\frobTosem{\cdot}_{c}$ and the cospan $d^z \tr{h_d} d^r \tl{q_d} d^n$ is in the image of $\frobTosem{\cdot}_{d}$. Thus they are both in the image of $\copair{\frobTosem{\cdot}_{c},\frobTosem{\cdot}_{d}}$. Concerning the two outermost cospans in \eqref{eq:decompFrob}, they are also in the image of $\copair{\frobTosem{\cdot}_{c},\frobTosem{\cdot}_{d}}$, as this is a morphism of $\col$-coloured prop and thus preserves and reflects the symmetry structure. Therefore, the whole of \eqref{eq:decompFrob} is in the image of $\copair{\frobTosem{\cdot}_{c},\frobTosem{\cdot}_{d}}$.\end{proof} 

Back to the main proof, thanks to the claim we have shown that the two outmost cospans of \eqref{eq:decomp} are in the image of $\copair{\frobTosem{\cdot}_{c},\frobTosem{\cdot}_{d}}$. Therefore $\gamma$ can be defined on $w \tr{f} G \tl{g} v$ by the values of $\synTosem{\cdot}$ and $\copair{\frobTosem{\cdot}_{c},\frobTosem{\cdot}_{d}}$ on its decomposition as in \eqref{eq:decomp}. This is a correctly and uniquely defined assignment: in the construction of decompositions \eqref{eq:decomp} and \eqref{eq:decompFrob}, the only variable parts are the different orderings that are picked for labelled nodes and for hyperedges in $\tilde{E}$, but these are immaterial since all the involved categories are symmetric monoidal.
\end{proof}



As a consequence of the claim enclosed in the proof of Theorem \ref{th:characterisation}, it is worth noticing that the Frobenius structure identifies the hypergraphs with no hyperedges, i.e. the sets of $\col$-labelled nodes.
\begin{corollary} There is an isomorphism of $\col$-coloured props $\frob{\col} \cong \FTerm{\emptyset,\col}$. \end{corollary}
As arrows of $\FTerm{\emptyset,\col}$ are the same thing as cospans in the slice category $\FINSET \setminus \col$, this corollary can be seen as a multi-sorted analogue of the well-known result reported in Proposition \ref{prop:cospanfrob}.

\paragraph{Double-pushout rewriting with interfaces} We conclude this section by recalling double-pushout (DPO) rewriting~\cite{HandbookDPO}, that we will use to compute in $\FTerm{\Sigma,\col}$. We will actually use a variation \cite{Gadducci1998} of the standard definition: instead of just rewriting an hypergraph $G$, we shall rewrite an hypergraph homomorphism $G \tl{} J$, standing for ``$G$ with interface $J$''. The intuition is that this form of computation retains the information that $J$ is how $G$ ``glues'' to a bigger context. This is needed both to match the syntactic notion of rewriting (Definition \ref{defn:rewprop}) and for ensuring decidability of confluence for terminating system (Section \ref{sec:confluence}). We formulate our definition at the level of \emph{adhesive categories} \cite{Lack2005}. This is the more abstract setting where DPO rewriting enjoys desirable properties (such as Church-Rosser) and where tools and algorithms for this form of rewriting are generally defined. 

\begin{definition}[DPOI rewriting] \label{defn:dpo}
Fix an adhesive category $\catA$. A \emph{rule for double-pushout rewriting with interfaces (DPOI rule)} is a span $L \tl{} J \tr{} R$  in $\catA$.
A \emph{DPOI rewriting system} $\RS$ is a finite set of DPOI rules.
Given morphisms $G \tl{r} I$ and $H \tl{p} I$ in $\catA$, we say that $G$ rewrites into $H$ via $\RS$ with interface $I$, notation $(G \tl{r} I) \DPOstep{\RS} (H \tl{p} I)$, if there exists a DPOI rule $L \tl{ } J \tr{} R$ in $\RS$ and a cospan
$J \tr{} C \tl{q} I$ such that the following diagram commutes and the two squares are pushouts. We call $L \to G$ a \emph{match} of $L$ in~$G$. 
\begin{equation}\label{eq:dpo2}
\raise25pt\hbox{$
\xymatrix@R=10pt@C=20pt{
L \ar[d]   & J \ar[d]
\ar@{}[dl]|(.8){\text{\large $\urcorner$}}
\ar@{}[dr]|(.8){\text{\large $\ulcorner$}}
\ar[l] \ar[r]  & R \ar[d] \\
 G &  C \ar[l] \ar[r]  & H \\
&  I \ar[u]^q \ar[ur]_p  \ar[ul]^r
}$}
\end{equation}
If $\catA$ has an initial object $0$, one can relax rewriting to act on cospans of the form $0 \tr{} G \tl{r} I$, seen as objects with interface $(G \tr{r} I)$. In this case we write $(0 \tr{} G \tl{r} I) \DPOstep{\RS} (0 \tr{} H \tl{p} I)$ for a rewriting step. 
\end{definition}

The following makes DPOI rewriting possible in $\catA = \Hyp{\Sigma,\col}$, see Ex. \ref{ex:DPOrewriting} below for an illustration.
\begin{proposition}\label{prop:adhesive} $\Hyp{\Sigma , \col}$ is an adhesive category. \end{proposition}
\begin{proof} 
The category $\Hyp{}$ of finite directed (unlabeled) hypergraphs is a presheaf category and thus adhesive \cite{Lack2005,BGKSZ-lics16}. The statement then follows because $\Hyp{\Sigma , \col}$ is defined as the slice category $\Hyp{} \setminus (\Sigma, \col)$ and adhesive categories are closed under slice~\cite{Lack2005}.
\end{proof}



\section{DPOI Implementation of Rewriting Modulo Frobenius} \label{sec:rewriting}

The category $\Hyp{\Sigma,\col}$ has an initial object $0$: the hypergraph with neither nodes nor hyperedges. Therefore, as mentioned in Definition \ref{defn:dpo}, we can equivalently think of hypergraphs with interface $(G \tr{} I)$ as cospans $0 \tr{} G \tl{} I$, and meaningfully define DPOI rewriting on the morphisms with source $0$ in $\FTerm{\Sigma,\col}$. 

On the other hand, our semantics $\allTosem{\cdot}$ maps diagrams of $\syntax{\Sigma}+\frob{\col}$ to cospans with any source. Thus, in order to interpret syntactic rewriting, we need an intermediate step where we ``fold'' the two interfaces $w_1, w_2$ of a string diagram $a \: w_1 \to w_2$ into one $w_1 w_2$. This is formally defined, with the help of the compact closed structure on $\syntax{\Sigma}+\frob{\col}$ (Proposition \ref{prop:compactclosed}), by an operation $\rewiring{\cdot}$:
\vspace{-.2cm}$$\cgr{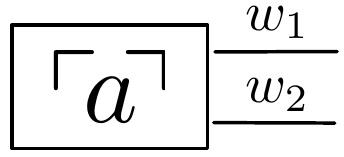} \quad \df \quad \cgr{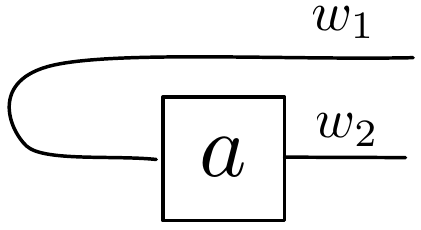}$$

We are now ready to formulate the correspondence theorem between syntactic rewriting in $\syntax{\Sigma}+\frob{\col}$ and DPOI rewriting in $\FTerm{\Sigma}$.

\begin{theorem}\label{thm:frobeniusrewriting}
Let $\rrule{l}{r}$ be any rewriting rule on $\syntax{\Sigma,\col}+\frob{\col}$. Then,
\[
a \Rightarrow_{\rrule{l}{r}} b  \quad \text{ iff } \quad \allTosem{\rewiring{a}} \DPOstep{\allTosem{\rrule{\rewiring{l}}{\rewiring{r}}}}  \allTosem{\rewiring{b}}\text{ .}
\]
\end{theorem}
\begin{proof}
On the direction from left to right, suppose that $a \Rightarrow_{\rrule{l}{r}} b$. Thus, by definition,
\begin{equation}\label{eq:proof1}
\cgr{diagA.pdf} = \cgr{rewl.pdf} \qquad \qquad \cgr{diagB.pdf} = \cgr{rewr.pdf}.
\end{equation}
Using the compact closed structure of $\syntax{\Sigma,\col}+\frob{\col}$ we can put $\rewiring{a}$ in the following shape
$$\cgr{rewiringA.pdf} = \cgr{rewiringA2.pdf} \eql{\eqref{eq:proof1}} \cgr{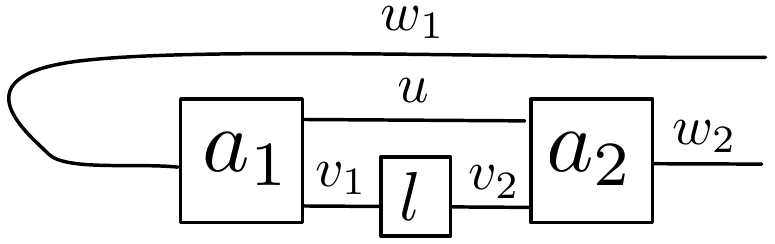} = \cgr{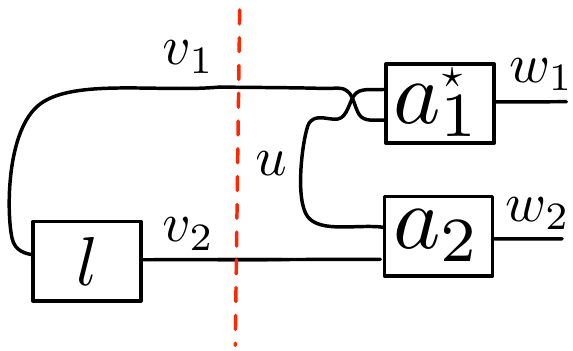}$$
The dashed line decomposes the rightmost diagram into $\rewiring{l} \: 0 \to v_1  v_2$ followed by a diagram of type $v_1 v_2 \to w_1 w_2$, which we name $\tilde{a}$. With analogous reasoning,
\begin{gather}\label{eq:rewiringcomp}  
\cgr{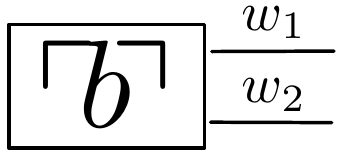} = \cgr{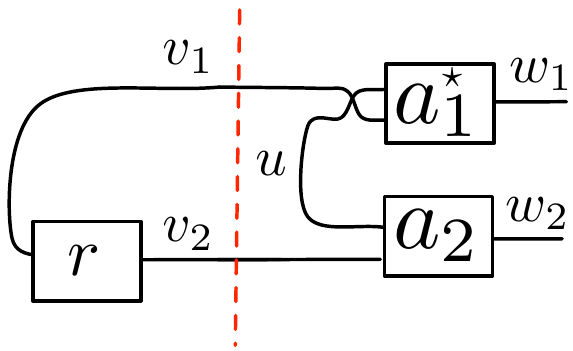}\qquad\text{ meaning that}\qquad \rewiring{a} = \rewiring{l} \poi \tilde{a} \text{ and }\rewiring{b} = \rewiring{r}\poi\tilde{a}.
\end{gather}
Next, we introduce cospans giving semantics to the various diagrams:
\begin{equation}\label{eq:defscospans}
\begin{aligned}
\allTosem{\rewiring{l}} = 0 \tr{} L \tl{} v_1v_2 \qquad \allTosem{\tilde{a}} = v_1v_2 \tr{} C \tl{} w_1 w_2 \qquad \allTosem{\rewiring{r}} = 0 \tr{} R \tl{} v_1v_2 \\
\allTosem{\rewiring{a}} = 0 \tr{} G \tl{} w_1 w_2 \qquad \allTosem{\rewiring{b}} = 0 \tr{} H \tl{} w_1 w_2.
\end{aligned}
\end{equation}
Equation \eqref{eq:rewiringcomp} tells that the cospan giving semantics to $\rewiring{a}$ (respectively, $\rewiring{b}$) is the composite of cospans giving semantics to $\rewiring{l}$ (respectively, $\rewiring{r}$) and $\tilde{a}$. As composition of cospans is by pushout, we obtain a double-pushout diagram as in \eqref{eq:dpo2} with $J = v_1v_2$ and $I = w_1w_2$, meaning that $\allTosem{\rewiring{a}} \DPOstep{\allTosem{\rrule{\rewiring{l}}{\rewiring{r}}}}  \allTosem{\rewiring{b}}$.
We now conclude the proof by showing the right to left direction of the statement. Suppose that $\allTosem{\rewiring{a}} \DPOstep{\allTosem{\rrule{\rewiring{l}}{\rewiring{r}}}}  \allTosem{\rewiring{b}}$. Naming cospans $\allTosem{\rewiring{a}}$, $\allTosem{\rewiring{b}}$, $\allTosem{\rewiring{l}}$ and $\allTosem{\rewiring{r}}$ as in \eqref{eq:defscospans}, this implies by definition the existence of a pushout complement $C$ yielding a DPOI diagram as \eqref{eq:dpo2} with $J = v_1v_2$ and $I = w_1w_2$. Now, pick $\hat{a} \: v_1 v_2 \to w_1 w_2$ such that $\allTosem{\hat{a}} = v_1 v_2 \tr{} C \tl{} w_1 w_2$, which exists by fullness of $\allTosem{\cdot}$. Because composition in $\FTerm{\Sigma,\col}$ is by pushout, the existence of such a DPOI diagram yields 
\begin{equation}\label{eq:secondproof}
\begin{aligned}
\allTosem{\rewiring{a}} = (0 \tr{} G \tl{} w_1 w_2) = (0 \tr{} L \tl{} v_1 v_2)\poi(v_1 v_2 \tr{} C \tl{} w_1 w_2) =   \allTosem{\rewiring{l}}\poi\allTosem{\hat{a}} \\ 
\allTosem{\rewiring{b}} = (0 \tr{} H \tl{} w_1 w_2) = (0 \tr{} R \tl{} v_1 v_2)\poi(v_1 v_2 \tr{} C \tl{} w_1 w_2) =   \allTosem{\rewiring{r}}\poi\allTosem{\hat{a}}.
\end{aligned}
\end{equation}
Because $\allTosem{\cdot}$ is faithful, \eqref{eq:secondproof} yields decompositions $\rewiring{a} = \rewiring{l}\poi\hat{a}$ and $\rewiring{b} = \rewiring{r}\poi\hat{a}$ also on the syntactic side. This allows for a rewriting step $a \Rew{\rrule{l}{r}} b$ as below, where the dashed lines show how the syntactic matching (\emph{cf.} the shape \eqref{eq:rewpropmatch}) is performed.
$$\cgr{diagA.pdf} = \cgr{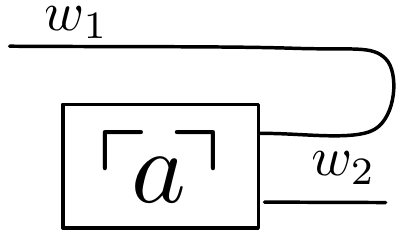} = \cgr{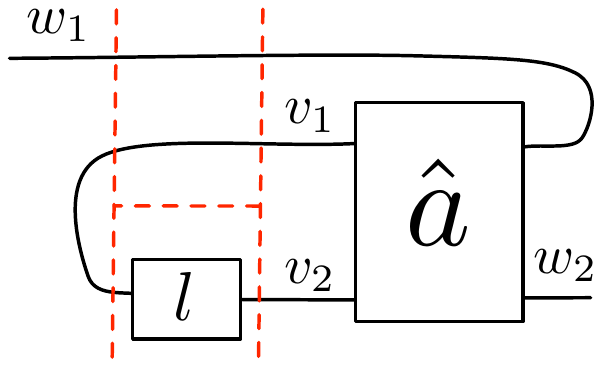} \Rew{\rrule{l}{r}} \cgr{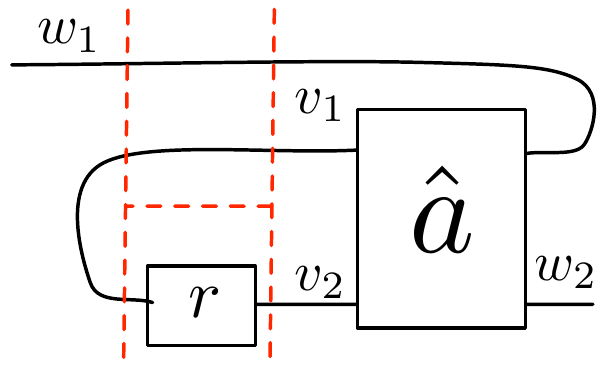} = \cgr{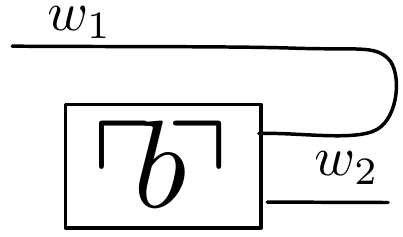} = \cgr{diagB.pdf}$$
\end{proof}
\begin{example}\label{ex:DPOrewriting} The syntactic rewriting step \eqref{eq:matchingproblem} takes place in $\syntax{\Sigma,\{c,d\}} + \frob{\{c,d\}}$ where $\Sigma = \{ o_1 \: c \to cd, o_2 \: dd \to c\}$. It maps via Theorem \ref{thm:frobeniusrewriting} to the following DPOI rewriting step in $\FTerm{\Sigma,\{c,d\}}$. We use numbers $1,2,3$ to disambiguate the ``folding'' of the rule interfaces.
\begin{equation*}
\xymatrix@C=15pt@R=10pt{
{\cgr{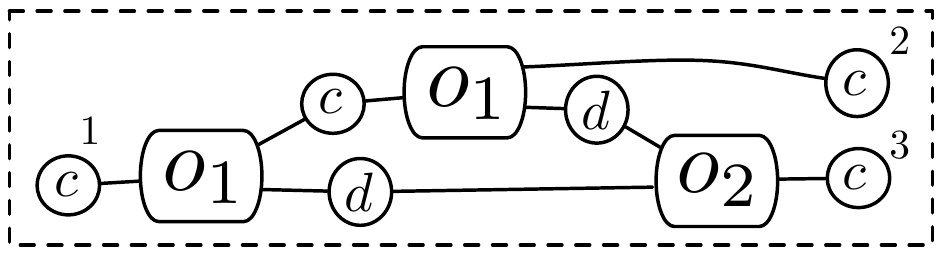}} \ar[d] & \ar@{}[dl]|(.55){\text{\large $\urcorner$}}
{\cgr{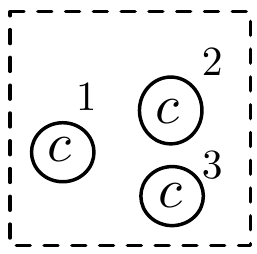}} \ar[d]_{} \ar[l] \ar[r]
\ar@{}[dr]|(.55){\text{\large $\ulcorner$}} & {\cgr{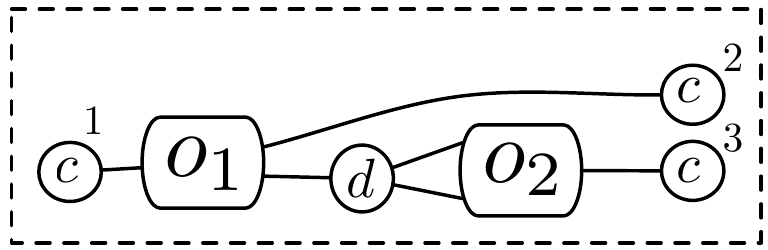}} \ar[d] \\
{\cgr{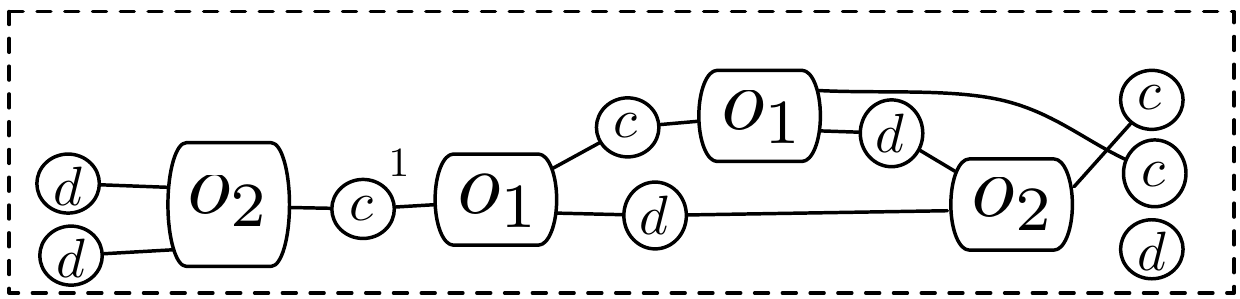}} & \ar[l] {\cgr{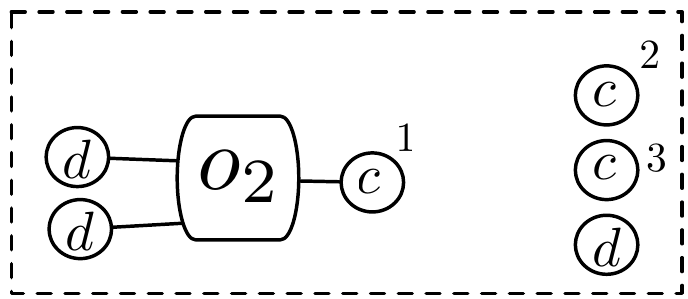}} \ar[r] & {\cgr{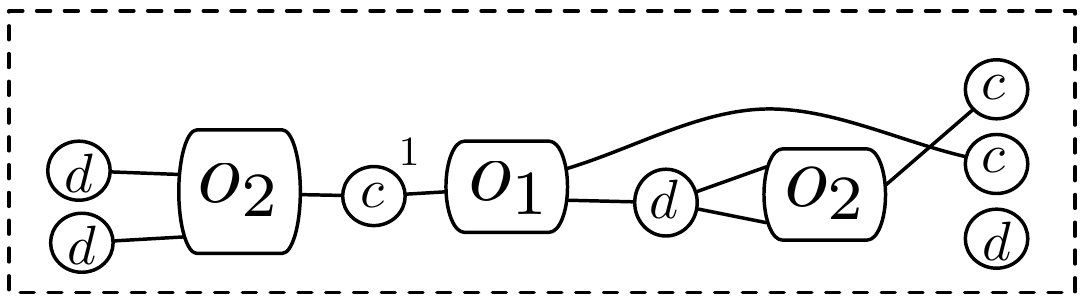}}
\\
& \ar[ul] \cgr{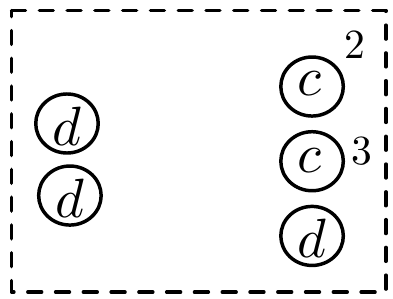} \ar[u] \ar[ur] &
}
\end{equation*}
\end{example}
\section{Decidability of Confluence}\label{sec:confluence}

This section verifies that the form of rewriting crystallised by Theorem \ref{thm:frobeniusrewriting} enjoys the \emph{Knuth-Bendix property}. Echoing the case of term rewriting, we use this terminology to mean that the confluence problem is reducible to critical pair analysis and both are decidable for terminating rewriting systems.

To this aim, we instantiate to our setting the results of \cite{BGKSZ-esop17}. There the author and collaborators showed that DPOI rewriting enjoys the aforementioned Knuth-Bendix property. Interfaces play a crucial role here, as Plump showed that for DPO rewriting (without interfaces) confluence is undecidable~\cite{Plump1993}.

\begin{definition}[DPOI Critical Pair]\label{def:critical}
Fix an adhesive category and DPOI rules $L_1 \tl{} K_1 \tr{} R_1$ and ${L_2 \tl{} K_2 \tr{} R_2}$. Consider the following two derivations with common source $S$.
\begin{equation}\label{eq:criticalpair}    
\vcenter{
\xymatrix@R=15pt{
    R_1 \ar[d] & \ar[l] K_1 \dlcorner  \ar[d] \ar[r] & L_1 \ar@{}[dr]|(.8){\text{\large $\ulcorner$}\qquad\quad} \ar@/^5pt/[dr]^{f_1} && \ar@/_5pt/[dl]_{f_2} \ar@{}[dl]|(.8){\qquad\quad\text{\large $\urcorner$}} L_2 & \ar[l] K_2 \ar[d] \drcorner\ar[r] & R_2 \ar[d] \\
    H_1  & \ar[l] C_1 \ar[rr] & & S  & & \ar[ll] C_2  \ar[r] & H_2 \\
    &&& \ar@/^/[ull] J \ar@{}[u]|{(\dagger)} \ar@/_/[urr] & & &
    }
}
\end{equation}
We say that $ (H_1\tl{}J) \DPOstepL{} (S \tl{}J ) \DPOstep{} (H_2\tl{}J)$ is a \emph{critical pair} if (i) $[f_1,f_2] \colon L_1 + L_2 \to S$ is epi and (ii) $(\dagger)$ is a pullback. 
It is \emph{joinable} if there exists $(W \tl{} J)$ such that $(H_1 \tl{} J) \DPOstep{}^{\star} (W  \tl{} J)\ {}^{\star}\!\!\!\DPOstepL{} \, (H_2 \tl{} J)$.
\end{definition}
Intuitively, condition (i) ensures that a critical pair $S$ is not bigger than $L_1+L_2$ and condition (ii) says that $J$ is the largest interface that allows both rewriting steps ($J$ is the ``intersection'' of $C_1$ and $C_2$). 

\begin{proposition}[\cite{BGKSZ-esop17}]\label{th:comput}
Suppose that $\catC$ satisfies the following assumptions: (1) it has an epi-mono factorisation system; (2) it has binary coproducts, pushouts and pullbacks; (3) it is adhesive (4) with all the pushouts stable under pullbacks. Then DPOI rewriting in $\catC$ has the Knuth-Bendix property for computable rewriting systems. 
\end{proposition}

\begin{remark} In the statement of Proposition \ref{th:comput}, computability refers to the conditions that (i) pullbacks are computable, (ii) for every pair of DPOI rules $L_1 \tl{} K_1 \tr{} R_1$ and ${L_2 \tl{} K_2 \tr{} R_2}$ the set of quotients of $L_1 + L_2$ is finite and computable, and (iii) for all $G \tl{} J$ one can compute every $H \tl{} I$ such that $(G \tl{} I) \DPOstep{} (H \tl{} I)$. 
In particular, (i)-(ii) ensure that the set of critical pairs is finite and computable, and (iii) ensures that any rewriting step is also computable--- see \cite{BGKSZ-esop17} for the full technical details. Caveats on computability are intended to single out those structures where it is reasonable to apply the DPOI mechanism, as opposed to e.g. systems with infinitely many critical pairs.
\end{remark}

\begin{remark} The one of Definition \ref{def:critical} is sometimes called a ``pre-critical'' pair, as no condition ensures that $L_1$ and $L_2$ actually overlap in $S$. This distinction can be formulated abstractly for DPO in adhesive categories when rules are (left- or right-) linear, see \cite{CorradiniStaf} for an overview. However, Proposition \ref{th:comput} works with arbitrary (non-linear) rules, \emph{cf.} \cite{BGKSZ-esop17}, whence the more general definition. Non-overlapping pairs can be singled out in our category of interest, $\FTerm{\Sigma,\col}$, and discarded for confluence testing as they are always joinable--- also, their set is finite, because any DPOI systems is computable in $\FTerm{\Sigma,\col}$. 
\end{remark}
\vspace{.2cm}
\begin{theorem}\label{th:confluenceDPOI} DPOI rewriting in $\FTerm{\Sigma,\col}$ has the Knuth-Bendix property. \end{theorem}
\vspace{-.2cm}
\begin{proof} Given that DPOI rewriting in $\FTerm{\Sigma,\col}$ is defined in terms of DPOI rewriting in $\Hyp{\Sigma,\col}$, it suffices to check the statement in $\Hyp{\Sigma,\col}$. We use Proposition \ref{th:comput}: assumptions (1)-(4) hold in any presheaf category and are closed under slice. Therefore, as $\Hyp{\Sigma,\col}$ is defined as the slice of a presheaf category (\emph{cf.} proof of Proposition \ref{prop:adhesive}), it satisfies (1)-(4). It remains to check that any DPOI rewriting system in $\Hyp{\Sigma,\col}$ is computable: the relevant observations are that in $\Hyp{\Sigma,\col}$ (i) pushouts and pushout complements are effectively computable and (ii) for any two rules there are only finitely many hypergraphs that may witness a critical pair as defined in Definition \ref{def:critical}.
\end{proof}



We would now wish to transfer Theorem \ref{th:confluenceDPOI} to syntactic rewriting in $\syntax{\Sigma,\col}+\frob{\col}$. This requires some extra care. In order to determine if a rewriting system $\RS$ on $\syntax{\Sigma,\col}+\frob{\col}$ is confluent, by Theorem~\ref{th:confluenceDPOI} and \ref{thm:frobeniusrewriting} it is enough that all the critical pairs in the DPOI system $\allTosem{\rewiring{\RS}}$ are joinable. However, for full decidability we also need to make sure that the converse holds: if one critical pair in $\allTosem{\rewiring{\RS}}$ is not joinable, then $\RS$ should not be confluent. To ensure this, we need to verify that all the critical pairs of $\allTosem{\mathcal{\rewiring{R}}}$ lay in the image of $\allTosem{\rewiring{\cdot}}$. This amounts to check that they all have discrete interfaces.

\vspace{.2cm}
\begin{lemma}\label{lemma:precritical-discrete} Consider a critical pair in $\Hyp{\Sigma,\col}$ as in~\eqref{eq:criticalpair}. 
If both $K_1$ and $K_2$ are discrete hypergraphs, so is the interface $J$.
\end{lemma}
\vspace{-.2cm}
\begin{proof}
For $i=1,2$, since $K_i$ is discrete, the hyperedges of $C_i$ are exactly those of $G_i$ that are not in $f_i(L_i)$.
Since $[f_1,f_2]\colon L_1+L_2 \to S$ is epi, all the hyperedges of $G$ are either in $f_1(L_1)$ or $f_2(L_2)$. Therefore, $J$ cannot contain any hyperedge.
\end{proof}

By definition of $\allTosem{\cdot}$, for every rule $L \tl{} K \tr{} R$ in $\allTosem{\rewiring{\RS}}$, $K$ is discrete. Therefore we can finally transfer the decidability result to the context of $\syntax{\Sigma,\col}+\frob{\col}$.

\vspace{.2cm}
\begin{corollary} Syntactic rewriting in $\syntax{\Sigma,\col}+\frob{\col}$ has the Knuth-Bendix property. \end{corollary}
\vspace{-.2cm}

\begin{proof} Since DPOI rewriting in $\FTerm{\Sigma,\col}$ has the Knuth-Bendix property (Theorem \ref{th:confluenceDPOI}) and the two forms of rewriting coincide (Theorem \ref{thm:frobeniusrewriting}), as discussed above it suffices to check that for a given critical pair in $\FTerm{\Sigma,\col}$, say witnessed by $(S\tl{}J )$, there exists $a$ in $\syntax{\Sigma}+\frob{\col}$, such that $\allTosem{a} = 0 \tr{} S\tl{}J$. As $\allTosem{\cdot}$ is full on $\FTerm{\Sigma,\col}$, it suffices that $J$ is discrete, which is true by Lemma \ref{lemma:precritical-discrete}.\pagebreak[3]\end{proof}

\section{Conclusions}

We described a sound and complete interpretation of string diagram rewriting in hypergraph categories as double-pushout rewriting of hypergraphs, and showed that it enjoys decidability of confluence for terminating rewriting systems. A chief advantage of this approach is that the challenge posed by performing matching modulo Frobenius equations disappears in the combinatorial model. This becomes important when studying axiomatisations with multiple Frobenius monoids: these can now be all seen as structural equations and baked into the combinatorial model, thus confining questions of confluence and termination to the non-Frobenius axioms. This application of our theory, which we plan to explore in future work, was the main reason to generalise the framework of \cite{BGKSZ-lics16,BGKSZ-esop17}, which is only able to absorb a single Frobenius structure. Another promising direction is the algebraic study of bipartite graphs (Example~\ref{ex:bipartite2}), which may be relevant for analysing diagrammatic languages, like biological metabolic networks \cite{MetabolicNetworks}, based on these structures.
\\[.5em]
\textbf{Acknowledgements} Thanks to Filippo Bonchi, Brendan Fong, Fabio Gadducci, Aleks Kissinger, the GAM participants and referees for useful comments and discussion on the topics of this paper.
\vspace{-.2cm}
{
\bibliographystyle{eptcs}
\bibliography{catBib3-3.bib}

\begin{thebibliography}{10}
\providecommand{\bibitemdeclare}[2]{}
\providecommand{\surnamestart}{}
\providecommand{\surnameend}{}
\providecommand{\urlprefix}{Available at }
\providecommand{\url}[1]{\texttt{#1}}
\providecommand{\href}[2]{\texttt{#2}}
\providecommand{\urlalt}[2]{\href{#1}{#2}}
\providecommand{\doi}[1]{doi:\urlalt{http://dx.doi.org/#1}{#1}}
\providecommand{\bibinfo}[2]{#2}

\bibitemdeclare{article}{BaezErbele-CategoriesInControl}
\bibitem{BaezErbele-CategoriesInControl}
\bibinfo{author}{John \surnamestart Baez\surnameend} \& \bibinfo{author}{Jason
  \surnamestart Erbele\surnameend} (\bibinfo{year}{2015}):
  \emph{\bibinfo{title}{Categories In Control}}.
\newblock {\sl \bibinfo{journal}{Theory and Application of Categories}}
  \bibinfo{volume}{30}, pp. \bibinfo{pages}{836--881}.

\bibitemdeclare{inproceedings}{BGKSZ-lics16}
\bibitem{BGKSZ-lics16}
\bibinfo{author}{Filippo \surnamestart Bonchi\surnameend},
  \bibinfo{author}{Fabio \surnamestart Gadducci\surnameend},
  \bibinfo{author}{Aleks \surnamestart Kissinger\surnameend},
  \bibinfo{author}{Pawel \surnamestart Soboci\'{n}ski\surnameend} \&
  \bibinfo{author}{Fabio \surnamestart Zanasi\surnameend}
  (\bibinfo{year}{2016}): \emph{\bibinfo{title}{Rewriting modulo symmetric
  monoidal structure}}.
\newblock In: {\sl \bibinfo{booktitle}{LiCS 2016}}, pp.
  \bibinfo{pages}{710--719}, \doi{10.1145/2933575.2935316}.

\bibitemdeclare{inproceedings}{BGKSZ-esop17}
\bibitem{BGKSZ-esop17}
\bibinfo{author}{Filippo \surnamestart Bonchi\surnameend},
  \bibinfo{author}{Fabio \surnamestart Gadducci\surnameend},
  \bibinfo{author}{Aleks \surnamestart Kissinger\surnameend},
  \bibinfo{author}{Pawel \surnamestart Soboci\'{n}ski\surnameend} \&
  \bibinfo{author}{Fabio \surnamestart Zanasi\surnameend}
  (\bibinfo{year}{2017}): \emph{\bibinfo{title}{Confluence of graph rewriting
  with interfaces}}.
\newblock In: {\sl \bibinfo{booktitle}{ESOP 2016}}, pp.
  \bibinfo{pages}{141--169}, \doi{10.1007/978-3-662-54434-1}.

\bibitemdeclare{inproceedings}{Bonchi2014b}
\bibitem{Bonchi2014b}
\bibinfo{author}{Filippo \surnamestart Bonchi\surnameend},
  \bibinfo{author}{Pawel \surnamestart Sobocinski\surnameend} \&
  \bibinfo{author}{Fabio \surnamestart Zanasi\surnameend}
  (\bibinfo{year}{2014}): \emph{\bibinfo{title}{A Categorical Semantics of
  Signal Flow Graphs}}.
\newblock In: {\sl \bibinfo{booktitle}{CONCUR 2014}}, {\sl
  \bibinfo{series}{LNCS}} \bibinfo{volume}{8704},
  \bibinfo{publisher}{Springer}, pp. \bibinfo{pages}{435--450},
  \doi{10.1007/978-3-662-44584-6}.

\bibitemdeclare{article}{BonchiSZ17}
\bibitem{BonchiSZ17}
\bibinfo{author}{Filippo \surnamestart Bonchi\surnameend},
  \bibinfo{author}{Pawel \surnamestart Sobocinski\surnameend} \&
  \bibinfo{author}{Fabio \surnamestart Zanasi\surnameend}
  (\bibinfo{year}{2017}): \emph{\bibinfo{title}{The Calculus of Signal Flow
  Diagrams {I:} Linear relations on streams}}.
\newblock {\sl \bibinfo{journal}{Inf. Comput.}} \bibinfo{volume}{252}, pp.
  \bibinfo{pages}{2--29}, \doi{10.1016/j.ic.2016.03.002}.

\bibitemdeclare{inproceedings}{Bruni01somealgebraic}
\bibitem{Bruni01somealgebraic}
\bibinfo{author}{Roberto \surnamestart Bruni\surnameend} \&
  \bibinfo{author}{Fabio \surnamestart Gadducci\surnameend}
  (\bibinfo{year}{2001}): \emph{\bibinfo{title}{Some algebraic laws for
  spans}}.
\newblock {\sl \bibinfo{series}{ENTCS}}~\bibinfo{volume}{44}, pp.
  \bibinfo{pages}{175--193}, \doi{10.1016/S1571-0661(04)80937-X}.

\bibitemdeclare{article}{Bruni2006}
\bibitem{Bruni2006}
\bibinfo{author}{Roberto \surnamestart Bruni\surnameend}, \bibinfo{author}{Ivan
  \surnamestart Lanese\surnameend} \& \bibinfo{author}{Ugo \surnamestart
  Montanari\surnameend} (\bibinfo{year}{2006}): \emph{\bibinfo{title}{A basic
  algebra of stateless connectors}}.
\newblock {\sl \bibinfo{journal}{Theoretical Computer Science}}
  \bibinfo{volume}{366}(\bibinfo{number}{1--2}), pp. \bibinfo{pages}{98--120},
  \doi{10.1016/j.tcs.2006.07.005}.

\bibitemdeclare{article}{Burroni1993}
\bibitem{Burroni1993}
\bibinfo{author}{Albert \surnamestart Burroni\surnameend}
  (\bibinfo{year}{1993}): \emph{\bibinfo{title}{Higher dimensional word
  problems with applications to equational logic}}.
\newblock {\sl \bibinfo{journal}{Theoretical Computer Science}}
  \bibinfo{volume}{115}(\bibinfo{number}{1}), pp. \bibinfo{pages}{43--62},
  \doi{10.1016/0304-3975(93)90054}.

\bibitemdeclare{article}{Carboni1987}
\bibitem{Carboni1987}
\bibinfo{author}{Aurelio \surnamestart Carboni\surnameend} \&
  \bibinfo{author}{R.~F.~C. \surnamestart Walters\surnameend}
  (\bibinfo{year}{1987}): \emph{\bibinfo{title}{Cartesian Bicategories {I}}}.
\newblock {\sl \bibinfo{journal}{Journal of Pure and Applied Algebra}}
  \bibinfo{volume}{49}(\bibinfo{number}{1-2}), pp. \bibinfo{pages}{11--32},
  \doi{10.1016/0022-4049(87)90121}.

\bibitemdeclare{inproceedings}{Coecke2008}
\bibitem{Coecke2008}
\bibinfo{author}{Bob \surnamestart Coecke\surnameend} \& \bibinfo{author}{Ross
  \surnamestart Duncan\surnameend} (\bibinfo{year}{2008}):
  \emph{\bibinfo{title}{Interacting Quantum Observables}}.
\newblock In: {\sl \bibinfo{booktitle}{ICALP 2008}}, {\sl
  \bibinfo{series}{LNCS}} \bibinfo{volume}{5216},
  \bibinfo{publisher}{Springer}, pp. \bibinfo{pages}{298--310},
  \doi{10.1007/978-3-540-70583-3}.

\bibitemdeclare{incollection}{HandbookDPO}
\bibitem{HandbookDPO}
\bibinfo{author}{A.~\surnamestart Corradini\surnameend},
  \bibinfo{author}{U.~\surnamestart Montanari\surnameend},
  \bibinfo{author}{F.~\surnamestart Rossi\surnameend},
  \bibinfo{author}{H.~\surnamestart Ehrig\surnameend},
  \bibinfo{author}{R.~\surnamestart Heckel\surnameend} \&
  \bibinfo{author}{M.~\surnamestart Loewe\surnameend} (\bibinfo{year}{1997}):
  \emph{\bibinfo{title}{Algebraic Approaches to Graph Transformation, Part I:
  Basic Concepts and Double Pushout Approach}}.
\newblock In: {\sl \bibinfo{booktitle}{Handbook of Graph Grammars}},
  \bibinfo{publisher}{University of Pisa}, pp. \bibinfo{pages}{163--246}.

\bibitemdeclare{inproceedings}{CorradiniStaf}
\bibitem{CorradiniStaf}
\bibinfo{author}{Andrea \surnamestart Corradini\surnameend}
  (\bibinfo{year}{2016}): \emph{\bibinfo{title}{On the definition of parallel
  independence in the algebraic approaches to graph transformation}}.
\newblock In: {\sl \bibinfo{booktitle}{STAF 2016}}, {\sl
  \bibinfo{series}{LNCS}} \bibinfo{volume}{9946},
  \bibinfo{publisher}{Springer}, \doi{10.1007/978-3-319-50230-4}.

\bibitemdeclare{inproceedings}{EhrigHPP04}
\bibitem{EhrigHPP04}
\bibinfo{author}{Hartmut \surnamestart Ehrig\surnameend},
  \bibinfo{author}{Annegret \surnamestart Habel\surnameend},
  \bibinfo{author}{Julia \surnamestart Padberg\surnameend} \&
  \bibinfo{author}{Ulrike \surnamestart Prange\surnameend}
  (\bibinfo{year}{2004}): \emph{\bibinfo{title}{Adhesive High-Level Replacement
  Categories and Systems}}.
\newblock In: {\sl \bibinfo{booktitle}{ICGT 2004}}, {\sl
  \bibinfo{series}{LNCS}} \bibinfo{volume}{2987},
  \bibinfo{publisher}{Springer}, pp. \bibinfo{pages}{144--160},
  \doi{10.1007/978-3-540-30203-2}.

\bibitemdeclare{conference}{Ehrig2004}
\bibitem{Ehrig2004}
\bibinfo{author}{Hartmut \surnamestart Ehrig\surnameend} \&
  \bibinfo{author}{Barbara \surnamestart K\"{o}nig\surnameend}
  (\bibinfo{year}{2004}): \emph{\bibinfo{title}{Deriving Bisimulation
  Congruences in the {DPO} Approach to Graph Rewriting}}.
\newblock In: {\sl \bibinfo{booktitle}{FoSSaCS 2004}}, {\sl
  \bibinfo{series}{LNCS}} \bibinfo{volume}{2987},
  \bibinfo{publisher}{Springer}, pp. \bibinfo{pages}{151--166},
  \doi{10.1007/978-3-540-24727-2}.

\bibitemdeclare{phdthesis}{Fon16}
\bibitem{Fon16}
\bibinfo{author}{Brendan \surnamestart Fong\surnameend} (\bibinfo{year}{2016}):
  \emph{\bibinfo{title}{The Algebra of Open and Interconnected Systems}}.
\newblock Ph.D. thesis, \bibinfo{school}{University of Oxford}.

\bibitemdeclare{inproceedings}{FZ-calco}
\bibitem{FZ-calco}
\bibinfo{author}{Brendan \surnamestart Fong\surnameend} \&
  \bibinfo{author}{Fabio \surnamestart Zanasi\surnameend}
  (\bibinfo{year}{2017}): \emph{\bibinfo{title}{A Universal construction for
  (co)relations}}.
\newblock In: {\sl \bibinfo{booktitle}{Proceedings of {CALCO}'17}}.

\bibitemdeclare{conference}{Gadducci1998}
\bibitem{Gadducci1998}
\bibinfo{author}{Fabio \surnamestart Gadducci\surnameend} \&
  \bibinfo{author}{Reiko \surnamestart Heckel\surnameend}
  (\bibinfo{year}{1997}): \emph{\bibinfo{title}{An inductive view of graph
  transformation}}.
\newblock In: {\sl \bibinfo{booktitle}{WADT 1997}}, {\sl
  \bibinfo{series}{LNCS}} \bibinfo{volume}{1376},
  \bibinfo{publisher}{Springer}, pp. \bibinfo{pages}{223--237},
  \doi{10.1007/3-540-64299-4}.

\bibitemdeclare{article}{HackneyColouredPROPs15}
\bibitem{HackneyColouredPROPs15}
\bibinfo{author}{Philip \surnamestart Hackney\surnameend} \&
  \bibinfo{author}{Marcy \surnamestart Robertson\surnameend}
  (\bibinfo{year}{2015}): \emph{\bibinfo{title}{On the Category of Props}}.
\newblock {\sl \bibinfo{journal}{Applied Categorical Structures}}
  \bibinfo{volume}{23}(\bibinfo{number}{4}), pp. \bibinfo{pages}{543--573},
  \doi{10.1007/s10485-014-9369-4}.

\bibitemdeclare{article}{KartsaklisSPC14}
\bibitem{KartsaklisSPC14}
\bibinfo{author}{Dimitri \surnamestart Kartsaklis\surnameend},
  \bibinfo{author}{Mehrnoosh \surnamestart Sadrzadeh\surnameend},
  \bibinfo{author}{Stephen \surnamestart Pulman\surnameend} \&
  \bibinfo{author}{Bob \surnamestart Coecke\surnameend} (\bibinfo{year}{2014}):
  \emph{\bibinfo{title}{Reasoning about Meaning in Natural Language with
  Compact Closed Categories and Frobenius Algebras}}.
\newblock {\sl \bibinfo{journal}{CoRR}} \bibinfo{volume}{abs/1401.5980}.
\newblock \urlprefix\url{http://arxiv.org/abs/1401.5980}.

\bibitemdeclare{conference}{Katis1997a}
\bibitem{Katis1997a}
\bibinfo{author}{Piergiulio \surnamestart Katis\surnameend},
  \bibinfo{author}{Nicoletta \surnamestart Sabadini\surnameend} \&
  \bibinfo{author}{Robert Frank~Carslaw \surnamestart Walters\surnameend}
  (\bibinfo{year}{1997}): \emph{\bibinfo{title}{{Span(Graph):} a categorical
  algebra of transition systems}}.
\newblock In: {\sl \bibinfo{booktitle}{Proceedings of AMAST '97}}, {\sl
  \bibinfo{series}{LNCS}} \bibinfo{volume}{1349},
  \bibinfo{publisher}{Springer}, pp. \bibinfo{pages}{322--336},
  \doi{10.1007/BFb0000479}.

\bibitemdeclare{incollection}{Street-2cats}
\bibitem{Street-2cats}
\bibinfo{author}{G.~M. \surnamestart Kelly\surnameend} \& \bibinfo{author}{Ross
  \surnamestart Street\surnameend} (\bibinfo{year}{1974}):
  \emph{\bibinfo{title}{Review of the elements of {$2$}-categories}}.
\newblock In: {\sl \bibinfo{booktitle}{Category {S}eminar ({P}roc. {S}em.,
  {S}ydney, 1972/1973)}}, \bibinfo{publisher}{Springer}, pp.
  \bibinfo{pages}{75--103. Lecture Notes in Math., Vol. 420},
  \doi{10.1016/0022-4049(72)90019-9}.

\bibitemdeclare{article}{KissingerHypergraph}
\bibitem{KissingerHypergraph}
\bibinfo{author}{Aleks \surnamestart Kissinger\surnameend}
  (\bibinfo{year}{2014}): \emph{\bibinfo{title}{Finite matrices are complete
  for (dagger-)hypergraph categories}}.
\newblock {\sl \bibinfo{journal}{CoRR}} \bibinfo{volume}{abs/1406.5942}.
\newblock \urlprefix\url{http://arxiv.org/abs/1406.5942}.

\bibitemdeclare{article}{Lack2004a}
\bibitem{Lack2004a}
\bibinfo{author}{Stephen \surnamestart Lack\surnameend} (\bibinfo{year}{2004}):
  \emph{\bibinfo{title}{Composing {PROPs}}}.
\newblock {\sl \bibinfo{journal}{Theory and Application of Categories}}
  \bibinfo{volume}{13}(\bibinfo{number}{9}), pp. \bibinfo{pages}{147--163}.

\bibitemdeclare{article}{Lack2005}
\bibitem{Lack2005}
\bibinfo{author}{Stephen \surnamestart Lack\surnameend} \&
  \bibinfo{author}{Pawe{\l} \surnamestart Soboci\'{n}ski\surnameend}
  (\bibinfo{year}{2005}): \emph{\bibinfo{title}{Adhesive and quasiadhesive
  categories}}.
\newblock {\sl \bibinfo{journal}{Theoretical Informatics and Applications}}
  \bibinfo{volume}{39}(\bibinfo{number}{3}), pp. \bibinfo{pages}{511--546},
  \doi{10.1051/ita:2005028}.

\bibitemdeclare{article}{MacDonald2009}
\bibitem{MacDonald2009}
\bibinfo{author}{John \surnamestart MacDonald\surnameend} \&
  \bibinfo{author}{Laura \surnamestart Scull\surnameend}
  (\bibinfo{year}{2009}): \emph{\bibinfo{title}{Amalgamations of categories}}.
\newblock {\sl \bibinfo{journal}{Can Math B}} \bibinfo{volume}{52}, pp.
  \bibinfo{pages}{273--284}, \doi{10.4153/CMB-2009-030-5}.

\bibitemdeclare{article}{MG17}
\bibitem{MG17}
\bibinfo{author}{Dan \surnamestart Marsden\surnameend} \&
  \bibinfo{author}{Fabrizio \surnamestart Genovese\surnameend}
  (\bibinfo{year}{2017}): \emph{\bibinfo{title}{Custom hypergraph categories
  via generalized relations}}.
\newblock {\sl \bibinfo{journal}{arXiv}} \bibinfo{volume}{abs/1703.01204}.
\newblock \urlprefix\url{http://arxiv.org/abs/1703.01204}.

\bibitemdeclare{inproceedings}{MimramFix}
\bibitem{MimramFix}
\bibinfo{author}{Samuel \surnamestart Mimram\surnameend}
  (\bibinfo{year}{2010}): \emph{\bibinfo{title}{Computing Critical Pairs in
  2-Dimensional Rewriting Systems}}.
\newblock In: {\sl \bibinfo{booktitle}{RTA 2010}}, {\sl
  \bibinfo{series}{LIPIcs}}~\bibinfo{volume}{6}, \bibinfo{publisher}{Schloss
  Dagstuhl - Leibniz-Zentrum f\"ur Informatik}, pp. \bibinfo{pages}{227--242},
  \doi{10.4230/LIPIcs.RTA.2010.227}.

\bibitemdeclare{article}{Pavlovic13}
\bibitem{Pavlovic13}
\bibinfo{author}{Dusko \surnamestart Pavlovic\surnameend}
  (\bibinfo{year}{2013}): \emph{\bibinfo{title}{Monoidal computer {I}: Basic
  computability by string diagrams}}.
\newblock {\sl \bibinfo{journal}{Information and Computation}}
  \bibinfo{volume}{226}, pp. \bibinfo{pages}{94--116},
  \doi{10.1016/j.ic.2013.03.007}.

\bibitemdeclare{incollection}{Plump1993}
\bibitem{Plump1993}
\bibinfo{author}{Detlef \surnamestart Plump\surnameend} (\bibinfo{year}{1993}):
  \emph{\bibinfo{title}{Hypergraph Rewriting: Critical Pairs and Undecidability
  of Confluence}}.
\newblock In: {\sl \bibinfo{booktitle}{Term Graph Rewriting: Theory and
  Practice}}, \bibinfo{publisher}{Wiley}, pp. \bibinfo{pages}{201--213}.

\bibitemdeclare{inproceedings}{Plump10}
\bibitem{Plump10}
\bibinfo{author}{Detlef \surnamestart Plump\surnameend} (\bibinfo{year}{2010}):
  \emph{\bibinfo{title}{Checking Graph-Transformation Systems for Confluence}}.
\newblock In: {\sl \bibinfo{booktitle}{Manipulation of Graphs, Algebras and
  Pictures}}, {\sl \bibinfo{series}{{ECEASST}}}~\bibinfo{volume}{26},
  \bibinfo{publisher}{EASST}.

\bibitemdeclare{article}{Rosebrugh2005}
\bibitem{Rosebrugh2005}
\bibinfo{author}{Robert \surnamestart Rosebrugh\surnameend},
  \bibinfo{author}{Nicoletta \surnamestart Sabadini\surnameend} \&
  \bibinfo{author}{R.~F.~C. \surnamestart Walters\surnameend}
  (\bibinfo{year}{2005}): \emph{\bibinfo{title}{Generic Commutative Separable
  Algebras and Cospans of Graphs}}.
\newblock {\sl \bibinfo{journal}{Theory and Application of Categories}}
  \bibinfo{volume}{17}(\bibinfo{number}{6}), pp. \bibinfo{pages}{164--177}.

\bibitemdeclare{article}{Selinger2009}
\bibitem{Selinger2009}
\bibinfo{author}{Peter \surnamestart Selinger\surnameend}
  (\bibinfo{year}{2011}): \emph{\bibinfo{title}{A survey of graphical languages
  for monoidal categories}}.
\newblock {\sl \bibinfo{journal}{Springer Lecture Notes in Physics}}
  \bibinfo{volume}{13}(\bibinfo{number}{813}), pp. \bibinfo{pages}{289--355}.

\bibitemdeclare{article}{MetabolicNetworks}
\bibitem{MetabolicNetworks}
\bibinfo{author}{\surnamestart Veeramani\surnameend},
  \bibinfo{author}{\surnamestart Balaji\surnameend},
  \bibinfo{author}{\surnamestart \surnameend} \& \bibinfo{author}{Joel~S
  \surnamestart Bader\surnameend} (\bibinfo{year}{2010}):
  \emph{\bibinfo{title}{Predicting Functional Associations from Metabolism
  Using Bi-Partite Network Algorithms}}.
\newblock {\sl \bibinfo{journal}{BMC Systems Biology}} \bibinfo{volume}{4},
  \doi{10.1186/1752-0509-4-95}.

\bibitemdeclare{phdthesis}{ZanasiThesis}
\bibitem{ZanasiThesis}
\bibinfo{author}{Fabio \surnamestart Zanasi\surnameend} (\bibinfo{year}{2015}):
  \emph{\bibinfo{title}{Interacting Hopf Algebras: the theory of linear
  systems}}.
\newblock Ph.D. thesis, \bibinfo{school}{{E}cole Normale Sup\'{e}rieure de
  Lyon}.

\bibitemdeclare{conference}{Zanasi16}
\bibitem{Zanasi16}
\bibinfo{author}{Fabio \surnamestart Zanasi\surnameend} (\bibinfo{year}{2016}):
  \emph{\bibinfo{title}{The Algebra of Partial Equivalence Relations}}.
\newblock In: {\sl \bibinfo{booktitle}{Mathematical Foundations of Program
  Semantics (MFPS)}}, \bibinfo{volume}{325}, pp. \bibinfo{pages}{313--333},
  \doi{10.1016/j.entcs.2016.09.046}.

\end{thebibliography}
}

\end{document}